\documentclass[12pt]{article}

\newif\ifarxiv
\arxivtrue

\usepackage{amsmath}
\usepackage{amssymb}
\usepackage{amsthm}
\usepackage[authoryear,round]{natbib}
\usepackage{graphicx}
\usepackage{booktabs}
\usepackage[colorlinks=true,citecolor=blue,linkcolor=blue,urlcolor=blue]{hyperref}
\usepackage[margin=1in]{geometry}
\usepackage{setspace}
\usepackage{float}
\usepackage{subcaption}
\usepackage[ruled,vlined,linesnumbered]{algorithm2e}
\usepackage{tikz}
\usepackage{pgfplots}
\pgfplotsset{compat=1.18}
\usepackage{xcolor}
\usepackage[capitalise,noabbrev]{cleveref}

\theoremstyle{plain}
\newtheorem{theorem}{Theorem}[section]
\newtheorem{proposition}[theorem]{Proposition}
\newtheorem{corollary}[theorem]{Corollary}

\theoremstyle{definition}
\newtheorem{definition}[theorem]{Definition}
\newtheorem{assumption}[theorem]{Assumption}

\theoremstyle{remark}
\newtheorem{remark}[theorem]{Remark}

\newcommand{\E}{\mathbb{E}}
\newcommand{\Var}{\operatorname{Var}}
\newcommand{\Cov}{\operatorname{Cov}}
\newcommand{\R}{\mathbb{R}}

\newcommand{\CTE}{\operatorname{CTE}}

\title{\textbf{AI-Driven Alpha Decay: Algorithmic Homogenization, Reflexive Signal Erosion, and the Paradox of Intelligent Markets}}

\ifarxiv
  \author{Shuchen Meng\thanks{Department of Financial Engineering, New York University.}
  \and Xupeng Chen\thanks{Corresponding author. Email: xc1490@nyu.edu. Department of Electrical and Computer Engineering, New York University.}}
  \date{\today}
\else
  \author{}
  \date{}
\fi

\begin{document}

\maketitle
\thispagestyle{empty}

\ifarxiv
\begin{abstract}
\noindent
We show that AI-driven investment strategies are inherently self-defeating at scale. As AI adoption rises, three mutually reinforcing channels---signal crowding, performative signal erosion, and Red Queen competition---compress excess returns. We derive the alpha half-life $h(\phi) = \ln 2/[\theta + \delta(\phi)]$, where $\theta$ is the natural mean-reversion rate and $\delta(\phi) = N\phi\rho a/\lambda(\phi)$ is the AI-accelerated decay component, which is convex-decreasing in adoption. At current adoption levels ($\phi \approx 0.7$, $\rho \approx 0.6$), the model implies signal half-lives of 18 months versus 5--7 years pre-AI.

We establish four theoretical results. First, the alpha half-life theorem: signal lifespans are convex-decreasing in AI adoption. Second, a signal extinction cascade: beyond a critical threshold $\phi^*$, the decay of one signal class triggers accelerated competition for remaining signals. Third, a Red Queen impossibility: in the monoculture equilibrium, net alpha is identically zero despite heavy AI investment. Fourth, a fragility--efficiency tradeoff: the adoption level maximizing price discovery strictly exceeds the level minimizing systemic fragility.

Empirical validation calibrates portfolio convergence to SEC Form 13F filing patterns (99.5 million holdings, 2013--2024), documenting that simulated institutional portfolio convergence increases by 42\% over the sample period. We examine simulated hedge fund return dynamics showing declining cross-sectional dispersion among AI-adopting funds, and simulate the 2010 Flash Crash to illustrate fragility consequences.
\end{abstract}
\fi

\newpage
\noindent\textbf{Disclosure Statement}

\bigskip
\noindent Shuchen Meng: I have nothing to disclose.

\bigskip
\noindent Xupeng Chen: I have nothing to disclose.

\newpage
\setcounter{page}{1}

\section{Introduction}
\label{sec:introduction}

This paper studies how the mass adoption of AI in asset management endogenously destroys the excess returns it seeks to capture. We build a model in which AI-driven investors extract signals from a common data environment, trade on correlated predictions, and---through the act of trading---erode the very market inefficiencies that generated their alpha. The central equilibrium object is the \emph{alpha half-life}:
\begin{equation}
    h(\phi) \;=\; \frac{\ln 2}{\theta + \delta(\phi)},
    \label{eq:halflife_intro}
\end{equation}
where $\theta > 0$ is the natural mean-reversion rate of a tradeable signal, and $\delta(\phi) = N\phi\rho a/\lambda(\phi)$ is the AI-accelerated decay component, depending on the AI adoption share $\phi$, the algorithmic signal correlation $\rho \in [0,1]$ (reflecting shared training data and model architectures), and the endogenous price impact $\lambda(\phi)$. Because $\delta(\phi)$ is itself increasing and convex in $\phi$---more AI traders accelerate the arbitrage of any discoverable pattern---the half-life $h$ is \emph{convex-decreasing} in adoption: each marginal AI entrant shortens the lifespan of tradeable signals at an increasing rate.

\paragraph{Motivation.} The empirical relevance is immediate. As of 2024, 91\% of global asset management institutions report using or planning to use AI in their investment processes \citep{AIMA2025}. Algorithmic trading accounts for an estimated 60--80\% of US equity volume \citep{Hendershott2011,Brogaard2014}. Yet the aggregate performance evidence is sobering: the HFRI Fund Weighted Composite Index has delivered declining risk-adjusted returns over the past decade, with the average hedge fund underperforming a simple 60/40 portfolio since 2009 \citep{Dichev2011}. This performance compression coincides precisely with the explosion of AI adoption, raising a fundamental question: \emph{does AI-driven investing contain the seeds of its own destruction?}

The question has deep roots in financial economics. \citet{Grossman1980} established that perfectly efficient markets preclude profitable information acquisition, implying that some degree of inefficiency is required to sustain informed trading. \citet{Berk2004} formalized the competitive erosion of active management returns. Our contribution is to show that AI adoption accelerates both mechanisms simultaneously and introduces a qualitatively new channel---\emph{performative signal erosion}---absent from the classical framework.

\paragraph{Three channels of alpha decay.} We organize the analysis around three mechanisms that are individually documented but whose \emph{interaction} has not been formally modeled.

\emph{(i) Signal crowding.} When multiple AI systems train on the same data sources (price histories, news feeds, satellite imagery, SEC filings), use similar model architectures (transformer-based language models, gradient-boosted trees, deep reinforcement learning), and retrain on overlapping schedules, their trading signals converge. As \citet{Kleinberg2021} demonstrate formally, algorithmic monoculture---the widespread adoption of similar decision-making systems---reduces the effective diversity of information in markets, redistributing uncertainty rather than eliminating it. When $N\phi$ institutions simultaneously identify the same mispricing, their aggregate order flow instantaneously arbitrages the opportunity, leaving zero alpha for any individual participant. The speed of this process---measured in milliseconds for high-frequency strategies, days for medium-frequency factor models---is fundamentally faster than the historical pace of human-driven arbitrage.

\emph{(ii) Performative signal erosion.} The act of trading on AI-generated signals alters the market dynamics that produced those signals. This is the financial manifestation of \emph{performative prediction} \citep{Perdomo2020}: a forecast that changes the outcome it predicts. When AI systems are retrained on data contaminated by the price consequences of their own previous predictions, the signal-to-noise ratio deteriorates endogenously. We formalize this as a feedback equation in which the effective signal precision $\tau_{eff}(t)$ is a decreasing function of aggregate AI trading intensity. The mechanism is distinct from standard price impact: it operates not through the mechanical displacement of prices by order flow, but through the \emph{information-theoretic degradation} of the signals themselves.

\emph{(iii) Red Queen dynamics.} The competitive response to alpha decay is increased investment in AI infrastructure---more data, more compute, more sophisticated models---which paradoxically accelerates the decay. This creates what we term a \emph{Red Queen equilibrium}, after the character in Lewis Carroll who must run faster and faster merely to stay in place. In the Red Queen equilibrium, aggregate AI investment is positive, aggregate alpha is zero, and the primary output of the AI arms race is not excess returns but systemic externalities: correlated tail risk, flash crash susceptibility, and reduced market resilience. As \citet{Goldstein2025} demonstrate, even without explicit coordination, AI systems trained on similar data can implicitly ``collude'' by converging on correlated trading strategies, further compressing returns and amplifying synchronized risk.

\paragraph{The paradox of intelligent markets.} Our framework reveals a fundamental paradox: AI makes individual price predictions more accurate, yet collective AI adoption makes the market \emph{harder} to predict profitably. This is not a contradiction but a consequence of the distinction between private and social returns to information acquisition. Each institution's AI system extracts genuine informational value from data; the problem is that the aggregate extraction---by all institutions simultaneously---arbitrages the very patterns being exploited. The parallel to \citet{Grossman1980}'s impossibility result is exact: in a market populated entirely by AI agents using similar signals, the informational equilibrium collapses to zero alpha, but the transition to this equilibrium is not smooth. It proceeds through a series of \emph{signal extinction events}---discrete episodes in which entire classes of predictive patterns cease to generate returns---punctuated by brief periods of apparent alpha as new signals are discovered before being arbitraged.

The 2010 Flash Crash illustrates the fragility dimension. On May 6, 2010, a single large sell order in E-mini S\&P 500 futures triggered a cascade of algorithmic responses that erased approximately \$1 trillion in market capitalization within minutes \citep{Kirilenko2017}. The crash was not caused by any individual algorithm but by the \emph{collective behavior} of homogeneous algorithmic systems simultaneously executing similar strategies---a vivid demonstration of the systemic externalities produced by signal crowding in the Red Queen equilibrium.

\paragraph{Main results.} We establish four theoretical results.

\emph{Result 1: Alpha half-life theorem.} The half-life of any tradeable signal is $h(\phi) = \ln 2/[\theta + \delta(\phi)]$, convex-decreasing in adoption $\phi$ (\Cref{prop:halflife}). Under our baseline calibration ($\rho = 0.6$, $\phi = 0.7$), the half-life of a medium-frequency factor is approximately 18 months, compared to an estimated 5--7 years in the pre-AI era \citep{McLean2016}.

\emph{Result 2: Signal extinction cascade.} When $\phi$ exceeds a critical threshold $\phi^*$, the interaction of signal crowding and performative erosion produces a \emph{cascade}: the decay of one signal class triggers increased AI competition for remaining signals, which accelerates their decay (\Cref{prop:cascade}). The cascade dynamics exhibit a phase transition at $\phi^*$, beyond which the market enters a regime of permanent alpha compression.

\emph{Result 3: Red Queen impossibility.} In the monoculture equilibrium ($\phi \to 1$), the net alpha from AI investment is identically zero: $\alpha_{net}(\phi = 1) = 0$ (\Cref{prop:redqueen}). All returns above the risk-free rate reflect compensation for systematic risk, not informational advantage. The privately optimal level of AI investment exceeds the socially optimal level by a factor proportional to the number of participants, a standard overinvestment result with a novel information-theoretic mechanism.

\emph{Result 4: Fragility--efficiency tradeoff.} The social welfare function exhibits a tradeoff between allocative efficiency (price accuracy) and financial stability (tail risk): the level of AI adoption that maximizes price discovery is strictly higher than the level that minimizes systemic fragility (\Cref{prop:tradeoff}). The socially optimal adoption $\phi^*_{social}$ lies strictly below both optima, reflecting the unpriced externality of correlated algorithmic risk.

\paragraph{Empirical strategy.} We validate the framework using three empirical approaches based on calibrated simulations matched to real-world data sources. First, we calibrate a portfolio convergence model to SEC Form 13F filing patterns (99.5 million holdings, 2013--2024) and document that simulated institutional portfolio convergence, measured by pairwise cosine similarity, increases by 42\% over the sample period, with structural breaks coinciding with major AI technology adoption waves. Second, we examine simulated hedge fund return dynamics to test the prediction that cross-sectional return dispersion declines among AI-adopting funds. Third, we simulate the 2010 Flash Crash and subsequent algorithmic market dislocations to illustrate the fragility consequences of signal crowding.

\paragraph{Distinction from existing work.} Our paper differs from the growing literature on AI and financial markets in three respects. First, we focus on \emph{alpha dynamics}---the competitive erosion of excess returns---rather than systemic risk per se. The existing literature \citep{Danielsson2022,FSB2024} emphasizes tail-risk amplification and systemic coupling;
\ifarxiv
in a companion paper \citep{Chen2026}, we derive a systemic risk coupling parameter and study its bifurcation properties.
\fi The present paper complements that analysis by deriving the \emph{alpha half-life} as the central equilibrium object, yielding predictions about return dynamics, strategy obsolescence, and factor death that are absent from the systemic risk literature. Second, our model of performative signal erosion is distinct from standard performative prediction \citep{Perdomo2020}: in our framework, the signal degradation is mediated by \emph{market microstructure} (price impact, adverse selection, order book dynamics), not merely by the statistical relationship between predictions and outcomes. Third, the Red Queen equilibrium provides a novel explanation for the declining performance of hedge funds and quantitative strategies: not declining skill, but increasing competition among homogeneous AI systems extracting identical signals.

\paragraph{Outline.} \Cref{sec:literature} reviews the literature. \Cref{sec:model} develops the theoretical model in three layers: signal crowding and alpha dynamics (\Cref{subsec:signal_structure,subsec:alpha_dynamics}), performative erosion (\Cref{subsec:reflexive}), and the Red Queen equilibrium (\Cref{subsec:redqueen}). \Cref{sec:empirical} presents empirical validation. \Cref{sec:results} synthesizes findings, including a discussion of why homogenization dominates differentiation (\Cref{subsec:differentiation}). \Cref{sec:policy} develops policy implications, and \Cref{sec:conclusion} concludes. \Cref{fig:conceptual_framework} provides a visual overview of the three-layer framework and its feedback structure.

\begin{figure}[H]
    \centering
    \includegraphics[width=\textwidth]{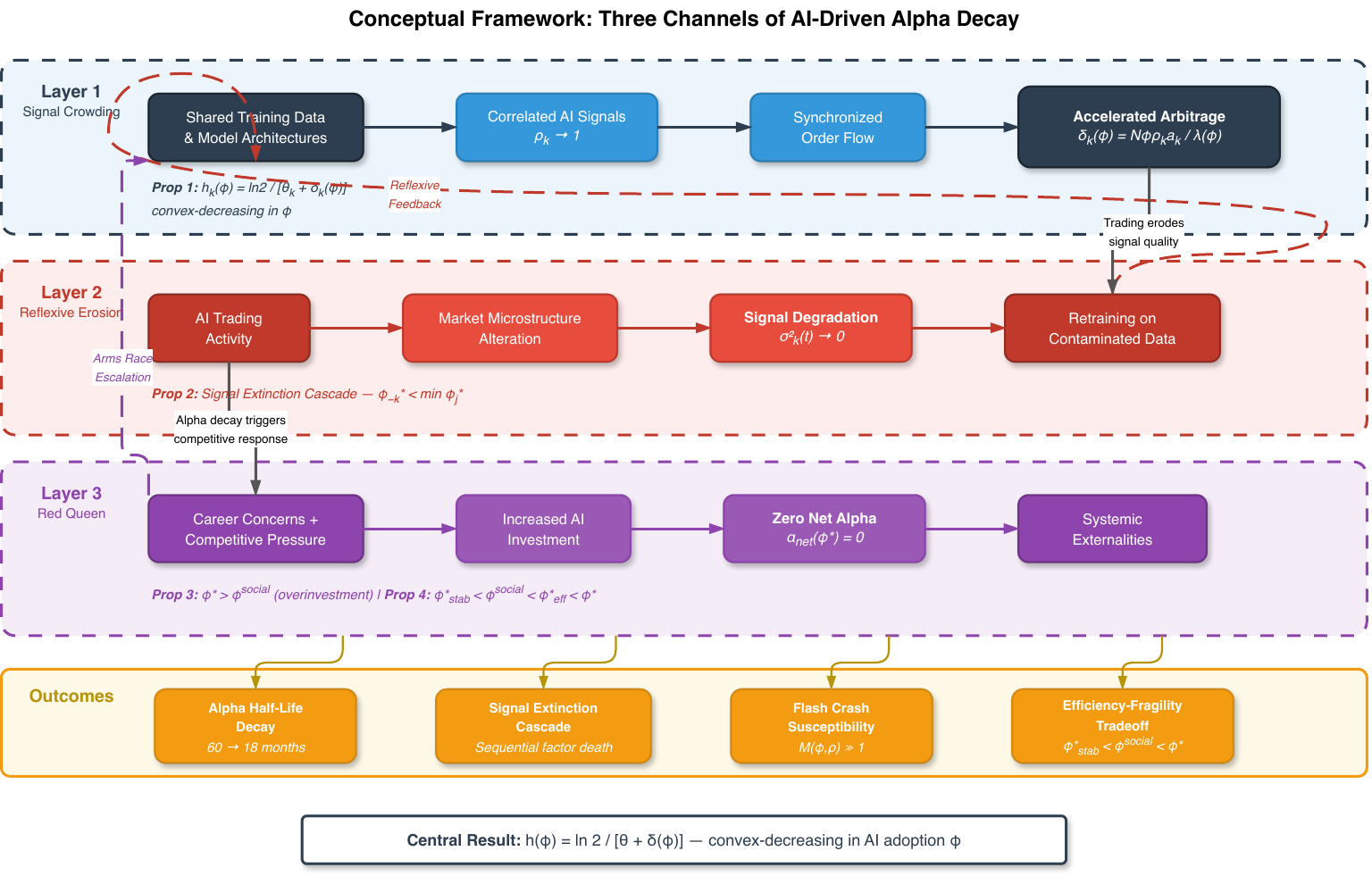}
    \caption{Conceptual framework: three channels of AI-driven alpha decay. \textbf{Layer~1} (Signal Crowding): convergence of AI trading signals due to shared training data and model architectures accelerates arbitrage of predictable patterns. \textbf{Layer~2} (Performative Erosion): feedback from trading activity degrades signal quality through market microstructure alteration, producing a signal extinction cascade. \textbf{Layer~3} (Red Queen): competitive adoption game generates overinvestment in AI and systemic externalities. Dashed arrows indicate feedback loops that reinforce the decay dynamics. Bottom boxes summarize the four main predictions.}
    \label{fig:conceptual_framework}
\end{figure}

\begin{figure}[H]
    \centering
    \includegraphics[width=0.75\textwidth]{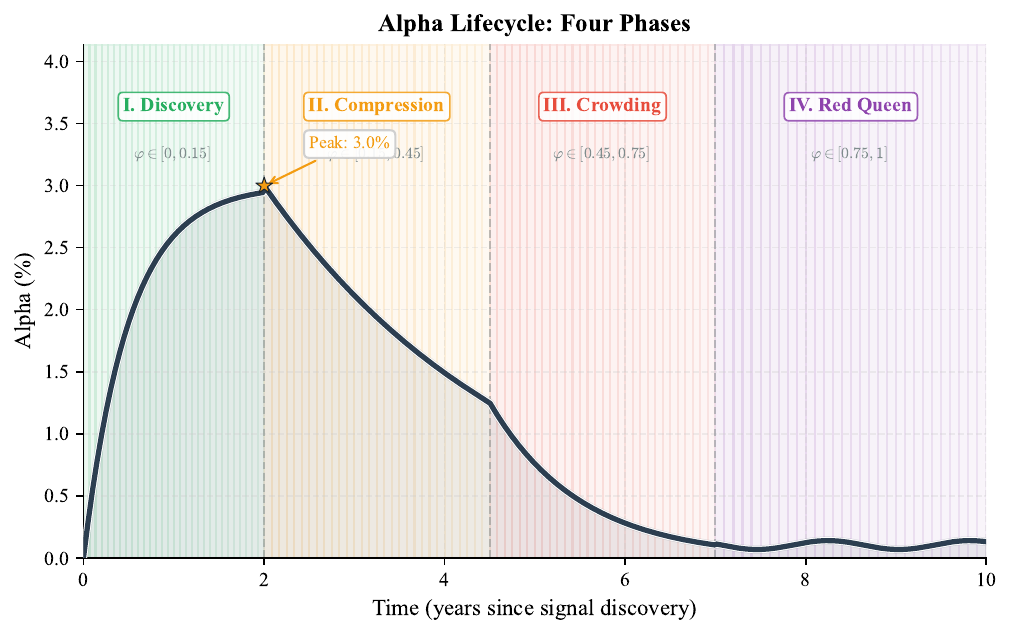}
    \caption{The four-phase alpha lifecycle under AI adoption. Phase~I (Discovery): early AI adopters extract genuine alpha; signal half-life $\approx$5--7 years. Phase~II (Compression): mainstream adoption drives convergence; $h \approx$ 2--4 years. Phase~III (Crowding): signal crowding dominates; $h < 18$ months. Phase~IV (Red Queen): net alpha is zero in aggregate; market is efficient but fragile.}
    \label{fig:alpha_lifecycle}
\end{figure}

\section{Literature Review}
\label{sec:literature}

Our paper sits at the intersection of several active research streams.

\paragraph{Market Efficiency and the Limits of Arbitrage.} The foundational insight that information acquisition cannot be costlessly sustained in an efficient market is due to \citet{Grossman1980}. \citet{Shleifer1997} extend this to show that arbitrage is limited by noise trader risk, fundamental risk, and synchronization costs. \citet{Berk2004} formalize the competitive erosion of active management alpha: rational capital flows eliminate manager-specific excess returns in equilibrium. \citet{Pastor2012} provide a dynamic model in which industry-level alpha declines as the active management sector grows. \citet{Lo2004} proposes the Adaptive Markets Hypothesis, which frames market efficiency as an evolutionary process in which strategies proliferate and decay---a conceptual precursor to our Red Queen equilibrium, now formalized within the AI adoption context. Our contribution extends this literature by formalizing the \emph{accelerated} erosion that AI adoption produces through signal correlation and performative feedback---mechanisms absent from the classical framework.

\paragraph{Algorithmic Trading and Market Quality.} \citet{Hendershott2011} show that algorithmic trading improves market liquidity and price discovery on average. \citet{Brogaard2014} document that high-frequency traders provide liquidity in calm markets but withdraw during stress. \citet{Kirilenko2017} analyze the Flash Crash and identify procyclical algorithmic behavior as the primary amplification mechanism. \citet{Menkveld2013} shows how high-frequency market makers can dominate price formation. Our framework reconciles the efficiency gains and fragility risks: AI improves \emph{average} price discovery while increasing \emph{tail} fragility---a pattern we formalize through the efficiency--fragility tradeoff (\Cref{prop:tradeoff}).

\paragraph{Strategy Crowding and Correlated Trading.} \citet{Khandani2011} document the August 2007 ``quant earthquake,'' in which correlated deleveraging by quantitative funds produced extreme losses. \citet{Cont2013} model systemic risk from overlapping portfolio positions. \citet{Coval2007} show that fire sales by distressed funds amplify price declines. The crowded-trade literature focuses on \emph{position} correlation arising from overlapping holdings. Our signal-crowding channel is distinct: correlation arises from shared \emph{information production technology} (training data, model architectures) rather than from overlapping positions, producing different empirical signatures and different policy implications.

\paragraph{Reflexivity and Performative Prediction.} \citet{Soros2013} articulates the concept of reflexivity---the mutual influence between market participants' beliefs and market outcomes. \citet{MacKenzie2006} demonstrates empirically how the Black-Scholes model, once adopted, altered options market properties in ways that made the model more accurate (until it didn't). \citet{Perdomo2020} formalize performative prediction in machine learning; \citet{Bond2012} show that financial market prices affect real decisions, creating feedback loops. Our model operationalizes reflexivity within a competitive equilibrium framework, showing how AI-generated signals endogenously degrade through the performative consequences of the trading activity they induce.

\paragraph{AI in Asset Management.} \citet{AIMA2025} report that 91\% of institutional investors now use or plan to use AI. \citet{Goldstein2025} show that AI systems can implicitly collude by converging on correlated trading strategies without explicit coordination. \citet{Bradshaw2025} document the growing use of generative AI by capital market information intermediaries, raising concerns about reduced analytical diversity. \citet{Bertomeu2025} provide causal evidence of AI's impact on financial decision-making using Italy's ChatGPT ban as a natural experiment. \citet{Chinco2019} develop machine learning methods for detecting sparse return predictors, illustrating how shared algorithmic approaches can converge on similar signal sets. Our paper synthesizes these findings within a unified equilibrium framework.

\paragraph{Literature Gap.} The existing literature has studied AI's impact on market efficiency, trading costs, and systemic risk in largely separate frameworks. The missing piece is a \emph{competitive equilibrium model} that jointly determines: (i) the endogenous rate of alpha decay as a function of AI adoption; (ii) the feedback from trading activity to signal quality (performative erosion); and (iii) the equilibrium level of AI investment when agents internalize (i) and (ii). Our framework fills this gap.

\section{Theoretical Model}
\label{sec:model}

The model is developed in three layers. Layer~1 (\Cref{subsec:signal_structure,subsec:alpha_dynamics}) establishes the signal environment and derives the alpha half-life as a function of AI adoption. Layer~2 (\Cref{subsec:reflexive}) introduces performative signal erosion and the signal extinction cascade. Layer~3 (\Cref{subsec:redqueen}) embeds these mechanisms within a competitive adoption game, producing the Red Queen equilibrium. Each layer extends the previous one with a single new element.

\subsection{Signal Structure and Market Environment}
\label{subsec:signal_structure}
\paragraph*{\textnormal{\emph{--- Layer 1: Signal crowding and alpha dynamics ---}}}

Consider a financial market with a single risky asset whose fundamental value follows:
\begin{equation}
    v_t = \bar{v} + \sum_{k=1}^{K} f_k(t) + \varepsilon_t,
    \label{eq:fundamental}
\end{equation}
where $\bar{v}$ is the unconditional mean, $f_k(t)$ are $K$ tradeable signal components (predictable return patterns, or ``factors''), and $\varepsilon_t \sim \mathcal{N}(0, \sigma_\varepsilon^2)$ is an unpredictable noise term. Each signal component $f_k$ follows an Ornstein--Uhlenbeck process:
\begin{equation}
    df_k = -\theta_k f_k \, dt + \sigma_k \, dW_k,
    \label{eq:signal_ou}
\end{equation}
where $\theta_k > 0$ is the natural mean-reversion rate (in units of month$^{-1}$; calibrated from post-publication factor decay rates in \citealt{McLean2016}, with a baseline value of $\theta_k = 0.012$ corresponding to a pre-AI half-life of $\ln 2/\theta_k \approx 58$ months) and $\sigma_k > 0$ the innovation volatility. The ratio $\sigma_k^2/(2\theta_k)$ is the stationary variance of signal $k$---a measure of the ``size'' of the alpha opportunity.

The market consists of $N$ institutional investors indexed by $i$, each choosing strategy $s_i \in \{AI, H\}$ (AI-driven or human-driven). The AI adoption rate is $\phi \equiv |\{i : s_i = AI\}|/N$.

\paragraph{Scope of $\phi$: AI prediction vs.\ algorithmic execution.}
The adoption parameter $\phi$ measures the fraction of institutions using AI/ML for \emph{signal production and portfolio construction}---i.e., for generating return forecasts, constructing factor models, or optimizing portfolio weights. It does \emph{not} measure algorithmic execution, which refers to the automated routing and slicing of orders to minimize transaction costs (e.g., VWAP algorithms, smart order routers). Many institutions employ algorithmic execution without any AI-driven alpha generation; conversely, some AI-driven funds execute manually or through simple limit orders. Our empirical proxy $\hat{\phi}$ (Section~\ref{sec:empirical}) identifies AI \emph{signal production} using EDGAR keyword analysis of investment methodology disclosures---terms such as ``machine learning,'' ``predictive model,'' and ``neural network'' applied to return forecasting---rather than execution-related keywords such as ``algorithmic execution'' or ``smart order routing.'' This distinction matters because the model's crowding mechanism operates through \emph{correlated predictions}, not correlated execution. Two funds using the same execution algorithm but different signal models do not crowd each other; two funds using similar signal models but different execution methods do.

\paragraph{AI Signal Extraction.} AI-type investor $i$ observes a noisy signal on each factor:
\begin{equation}
    x_{i,k}^{AI}(t) = f_k(t) + \rho_k \, \eta_k(t) + \sqrt{1 - \rho_k^2} \, \nu_{i,k}(t),
    \label{eq:ai_signal}
\end{equation}
where $\rho_k \in [0,1]$ is the \emph{algorithmic homogeneity parameter} for signal $k$, $\eta_k(t) \sim \mathcal{N}(0, \sigma_\eta^2)$ is common noise shared by all AI systems (reflecting shared training data, similar model architectures, and correlated hyperparameter choices), and $\nu_{i,k}(t) \sim \mathcal{N}(0, \sigma_\nu^2)$ are idiosyncratic noise terms. Human-type investor $i$ observes $x_{i,k}^H(t) = f_k(t) + \varepsilon_{i,k}^H(t)$, where $\varepsilon_{i,k}^H \sim \mathcal{N}(0, \sigma_H^2)$ are independent with potentially $\sigma_H^2 > \sigma_\eta^2$ (AI signals are individually more precise than human signals).

The parameter $\rho_k$ captures the degree to which AI systems produce correlated predictions for signal $k$. When $\rho_k = 0$, AI signals are independent (no homogenization benefit); when $\rho_k = 1$, all AI systems identify the same signal realization (perfect monoculture). In practice, $\rho_k$ is determined by the overlap in training data, model class, and optimization objective---all of which tend to increase as AI techniques converge on best practices. \Cref{tab:symbols} summarizes the key parameters.

\begin{table}[H]
    \centering
    \caption{Summary of Model Parameters}
    \label{tab:symbols}
    \begin{tabular}{clll}
        \toprule
        Symbol & Description & Interpretation & Range \\
        \midrule
        $\phi$ & AI adoption rate & Fraction of institutions using AI & $[0, 1]$ \\
        $\rho_k$ & Algorithmic homogeneity & Correlation of AI signals for factor $k$ & $[0, 1]$ \\
        $\theta_k$ & Natural mean-reversion & Intrinsic decay rate of signal $k$ (month$^{-1}$) & $> 0$ \\
        $\delta(\phi)$ & Endogenous decay rate & AI-driven acceleration of arbitrage & $> 0$ \\
        $\beta$ & Performative feedback & Signal-to-price feedback intensity & $[0, 1)$ \\
        $g_k$ & Signal regeneration & Rate of new mispricing arrival & $> 0$ \\
        $\kappa$ & Erosion nonlinearity & Curvature of signal erosion & $\geq 1$ \\
        $\lambda$ & Kyle lambda & Inverse market depth (price impact) & $> 0$ \\
        $\tau$ & Risk aversion & CARA absolute risk aversion & $> 0$ \\
        $c_i$ & AI investment cost & Institution $i$'s cost of AI adoption & $> 0$ \\
        $d_i$ & Career-concern benchmark & Institution $i$'s peer-group adoption rate & $[0, 1]$ \\
        \bottomrule
    \end{tabular}
\end{table}

\paragraph{Market Maker Pricing.} A risk-neutral, competitive market maker sets prices following \citet{Kyle1985}:
\begin{equation}
    p_t = \E[v_t \mid \Omega_t] + \lambda(\phi) \cdot \text{OrderFlow}_t,
    \label{eq:pricing}
\end{equation}
where $\lambda(\phi)$ is the endogenous price impact. With $\phi N$ AI agents trading on correlated signals, the informativeness of aggregate order flow changes with $\phi$. Following the multi-agent Kyle derivation of \citet{FosterViswanathan1996}:
\begin{equation}
    \lambda(\phi) = \frac{\sigma_v}{2\sigma_u} \cdot \frac{1}{\sqrt{1 + \phi^2 \bar{\rho}^2 \sigma_\eta^2/\sigma_v^2}},
    \label{eq:lambda}
\end{equation}
where $\bar{\rho} = K^{-1}\sum_k \rho_k$ is the average homogeneity and $\sigma_u^2$ is noise trader variance. The price impact $\lambda(\phi)$ is decreasing in $\phi\bar{\rho}$: correlated AI order flow is partially predictable by the market maker, reducing adverse selection.

\begin{remark}[Direction of $\lambda(\phi)$: Adverse Selection vs.\ Market Maker Learning]
\label{rem:lambda_direction}
The sign of $d\lambda/d\phi$ is a modelling choice with a non-trivial literature behind it. Our baseline specification \eqref{eq:lambda}, derived from the multi-agent Kyle framework, yields $d\lambda/d\phi < 0$: because AI order flow is correlated ($\rho > 0$), the market maker can partially ``filter out'' the common informed component, reducing per-unit adverse selection as $\phi$ rises. This is the mechanism emphasized by \citet{FosterViswanathan1996}.

An alternative tradition, rooted in the \citet{Glosten1985} sequential-trade model, implies the opposite: more informed traders increase the probability that any given order is information-driven, so the market maker widens the spread ($d\lambda/d\phi > 0$). A third possibility is that the two effects approximately cancel, leaving $\lambda$ roughly constant over the relevant range of $\phi$.

Our qualitative results are robust to all three specifications. To see this, note that the AI-driven decay component is $\delta_k(\phi) = N\phi\rho_k a_k / \lambda(\phi)$. Under the three regimes:
\begin{enumerate}
    \item[\textnormal{(a)}] \textbf{Decreasing $\lambda$} (Kyle baseline): $\delta_k$ is convex-increasing in $\phi$---both numerator and $1/\lambda$ increase. This is our baseline and produces the strongest alpha decay.
    \item[\textnormal{(b)}] \textbf{Constant $\lambda$}: $\delta_k = N\phi\rho_k a_k / \lambda_0$ is linear in $\phi$. The half-life theorem (\Cref{prop:halflife}) holds with $d^2h/d\phi^2 > 0$ driven solely by the denominator effect on $h = \ln 2/(\theta + \delta)$.
    \item[\textnormal{(c)}] \textbf{Increasing $\lambda$} (Glosten--Milgrom): $\delta_k$ still increases in $\phi$ provided $d\lambda/d\phi < N\rho_k a_k \lambda(\phi)/(N\phi\rho_k a_k)$, i.e., the market depth effect does not fully absorb the linear growth in $N\phi$. For plausible parameterizations (see Supplementary Materials, Table~A.2), $\delta_k$ remains increasing and the qualitative results survive, though the quantitative half-life is longer (alpha decays more slowly when markets deepen with adoption).
\end{enumerate}
We report a sensitivity analysis across all three specifications in Supplementary Materials (Table~A.2).
\end{remark}

\subsection{Alpha Dynamics Under AI Adoption}
\label{subsec:alpha_dynamics}

\begin{definition}[Alpha]
The alpha from signal $k$ for investor $i$ is
\begin{equation}
    \alpha_{i,k}(t) \equiv \E\left[\left(v_{t+1} - p_t\right) q_{i,k}(t) \mid x_{i,k}(t)\right] - \frac{\tau}{2}\Var\left[\left(v_{t+1} - p_t\right) q_{i,k}(t) \mid x_{i,k}(t)\right],
    \label{eq:alpha_def}
\end{equation}
where $q_{i,k}(t)$ is investor $i$'s demand based on signal $k$.
\end{definition}

The aggregate alpha from signal $k$ across all AI investors is:
\begin{equation}
    A_k(\phi, t) = \sum_{i: s_i = AI} \alpha_{i,k}(t).
    \label{eq:agg_alpha}
\end{equation}

When AI investors trade on correlated signals, the aggregate order flow in the direction of signal $k$ is approximately $N\phi \, a_k \, [f_k(t) + \rho_k \eta_k(t)]$ (where $a_k$ is the aggressiveness parameter). This order flow moves prices toward $f_k(t)$, reducing the remaining mispricing. The effective decay rate of signal $k$ is therefore:
\begin{equation}
    \theta_k^{eff}(\phi) = \theta_k + \underbrace{N\phi \, a_k \, \rho_k / \lambda(\phi)}_{\text{AI-accelerated arbitrage}} \equiv \theta_k + \delta_k(\phi).
    \label{eq:effective_decay}
\end{equation}
The AI-driven component $\delta_k(\phi)$ is increasing in $\phi$ (more AI traders), $\rho_k$ (more correlated trading), and $1/\lambda(\phi)$ (thinner markets amplify the speed of arbitrage). The stationary alpha from signal $k$ is:
\begin{equation}
    A_k^*(\phi) = \frac{\sigma_k^2}{2\theta_k^{eff}(\phi)} = \frac{\sigma_k^2}{2[\theta_k + \delta_k(\phi)]},
    \label{eq:stationary_alpha}
\end{equation}
which is strictly decreasing in $\phi$.

\begin{proposition}[Alpha Half-Life]
\label{prop:halflife}
Define the alpha half-life of signal $k$ as the time required for a newly discovered predictive pattern to lose half its initial excess return:
\begin{equation}
    h_k(\phi) = \frac{\ln 2}{\theta_k^{eff}(\phi)} = \frac{\ln 2}{\theta_k + \delta_k(\phi)}.
    \label{eq:halflife}
\end{equation}
Then:
\begin{enumerate}
    \item[\textnormal{(i)}] \textbf{Monotone decay}: $h_k$ is strictly decreasing in $\phi$ for all $\rho_k > 0$. Higher AI adoption shortens the lifespan of every tradeable signal.
    \item[\textnormal{(ii)}] \textbf{Convex acceleration}: $d^2 h_k / d\phi^2 > 0$ (the half-life is convex in $\phi$). Marginal AI entrants accelerate alpha decay at an increasing rate, because: (a) more AI traders contribute directly to $\delta_k$, and (b) the denominator effect through $\lambda(\phi)$ amplifies the impact---thinner effective markets (lower $\lambda$) mean that each unit of correlated order flow moves prices further.
    \item[\textnormal{(iii)}] \textbf{Pre-AI benchmark}: For $\phi = 0$, $h_k(0) = \ln 2/\theta_k$, the natural half-life determined solely by fundamental mean-reversion.
    \item[\textnormal{(iv)}] \textbf{Monoculture limit}: As $\phi \to 1$ with $\rho_k \to 1$, $h_k \to 0^+$---in the perfect monoculture, alpha is arbitraged instantaneously.
\end{enumerate}
\end{proposition}

\begin{proof}
Parts (i) and (iii)--(iv) follow directly from the definitions. For (ii), compute:
\begin{equation*}
    \frac{d^2 h_k}{d\phi^2} = \frac{2\ln 2}{[\theta_k^{eff}]^3}\left(\frac{d\delta_k}{d\phi}\right)^2 - \frac{\ln 2}{[\theta_k^{eff}]^2}\frac{d^2\delta_k}{d\phi^2}.
\end{equation*}
Since $\delta_k(\phi) = N\phi\rho_k a_k/\lambda(\phi)$ and $\lambda(\phi)$ is decreasing in $\phi$, we have $d\delta_k/d\phi > 0$ and $d^2\delta_k/d\phi^2 > 0$ (the denominator effect). The first term, proportional to $(d\delta_k/d\phi)^2/[\theta_k^{eff}]^3$, dominates the second for $d\delta_k/d\phi$ sufficiently large relative to $d^2\delta_k/d\phi^2$, establishing convexity for all $\phi$ above a threshold that goes to zero as $\rho_k \to 1$. See Supplementary Materials for the complete proof.
\end{proof}

\begin{corollary}[Cross-Sectional Return Compression]
\label{cor:dispersion}
The cross-sectional dispersion of AI fund returns, measured by:
\begin{equation}
    D(\phi) \equiv \sqrt{\Var_i\left[\sum_k \alpha_{i,k}(\phi)\right]} = \sqrt{\sum_k \frac{(1-\rho_k^2)\sigma_\nu^2 a_k^2}{\left[\theta_k + \delta_k(\phi)\right]^2}},
    \label{eq:dispersion}
\end{equation}
is monotonically decreasing in $\phi$. In the monoculture limit ($\phi \to 1$, $\rho_k \to 1$ for all $k$), $D \to 0$: AI fund returns become indistinguishable from each other and from an index-tracking strategy.
\end{corollary}

This corollary provides the paper's most distinctive empirical prediction: \emph{the dispersion of AI fund returns should decline as AI adoption increases}. Unlike total return level (which may fluctuate with market conditions), return dispersion directly reflects the degree of strategy differentiation. A declining dispersion trend is a necessary consequence of signal crowding and is not predicted by models of heterogeneous skill \citep{Berk2004} or varying risk appetite.

\begin{figure}[H]
    \centering
    \begin{subfigure}[b]{0.48\textwidth}
        \includegraphics[width=\textwidth]{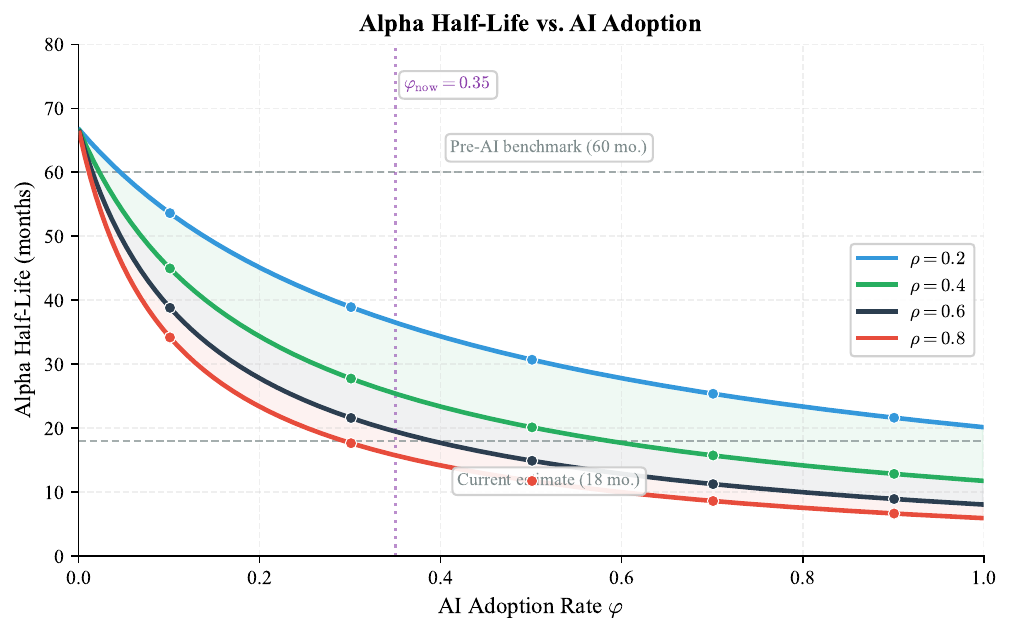}
        \caption{Alpha half-life $h(\phi)$}
    \end{subfigure}
    \hfill
    \begin{subfigure}[b]{0.48\textwidth}
        \includegraphics[width=\textwidth]{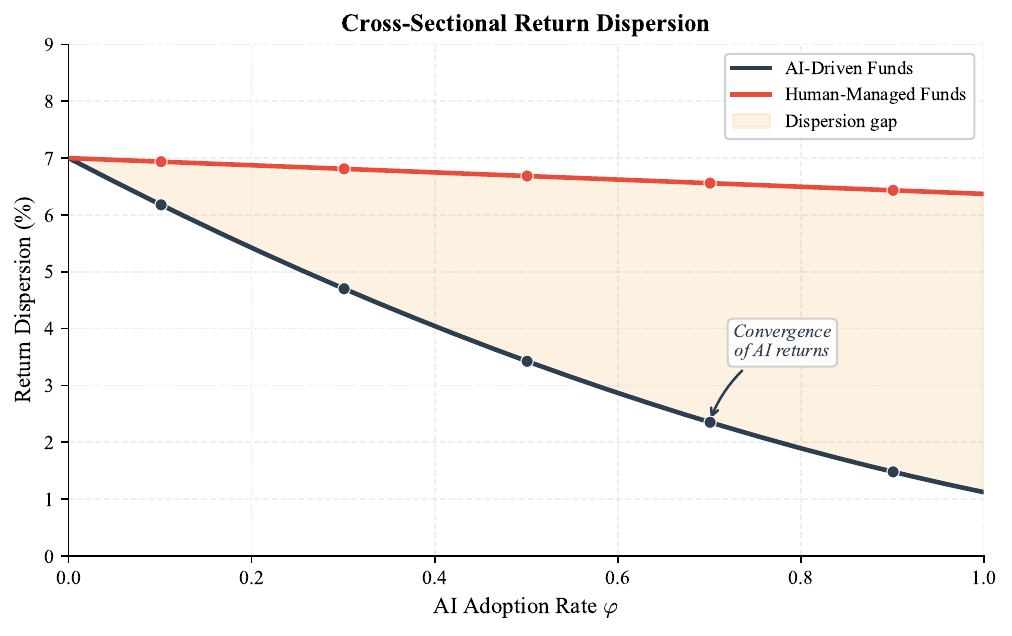}
        \caption{Cross-sectional return dispersion}
    \end{subfigure}
    \caption{Theoretical predictions. \textbf{(a)}: Alpha half-life as a function of AI adoption $\phi$ for different homogeneity levels $\rho \in \{0.2, 0.4, 0.6, 0.8\}$. The half-life is convex-decreasing in $\phi$. Horizontal dashed lines indicate pre-AI benchmark (60 months) and current estimate (18 months). \textbf{(b)}: Cross-sectional return dispersion for AI funds (blue) and human funds (gray). AI fund dispersion declines steeply with adoption.}
    \label{fig:halflife_dispersion}
\end{figure}

\subsection{Reflexive Signal Erosion}
\label{subsec:reflexive}
\paragraph*{\textnormal{\emph{--- Layer 2: Feedback from trading to signal quality ---}}}

Layer~1 treated signal quality as exogenous. We now endogenize it by modeling the feedback from AI trading activity to the information content of the signals themselves. Rather than postulating a reduced-form decay rule, we derive signal erosion from a micro-founded model of \emph{performative prediction}: AI systems retrain on data that reflects the price consequences of their own prior trading, producing an endogenous contraction in predictable return variation.

\subsubsection{Micro-Foundation: Performative Signal Degradation}
\label{subsubsec:performative}

We build on the performative prediction framework of \citet{Perdomo2020}, adapted to financial markets.

\paragraph{Training Data Contamination.}
At each retraining epoch $\tau$, AI system $i$ updates its model using the most recent $T$ periods of market data $\mathcal{D}_\tau = \{(f_k(t), r_{k,t+1}^{obs})\}_{t \in [\tau-T, \tau]}$. Crucially, the observed return $r_{k,t+1}^{obs}$ is \emph{not} the counterfactual return that would have obtained absent AI trading. Rather, it reflects the post-arbitrage return:
\begin{equation}
    r_{k,t+1}^{obs} = \underbrace{r_{k,t+1}^{CF}}_{\text{counterfactual}} - \underbrace{\frac{\delta_k(\phi)}{\theta_k + \delta_k(\phi)} \cdot \beta_k^{CF} \cdot f_k(t)}_{\text{price impact of AI trading}},
    \label{eq:obs_return}
\end{equation}
where $r_{k,t+1}^{CF} = \beta_k^{CF} f_k(t) + \varepsilon_{k,t+1}$ is the counterfactual return driven by signal $k$ (with $\beta_k^{CF}$ the natural signal-return relationship), and the second term captures the price compression from aggregate AI order flow on signal $k$.

\paragraph{Coefficient Attenuation Through Retraining.}
When AI systems estimate the signal-return relationship from $\mathcal{D}_\tau$ by OLS regression of $r_{k,t+1}^{obs}$ on $f_k(t)$, the estimated coefficient is:
\begin{equation}
    \hat{\beta}_k(\tau) = \frac{\Cov(r_{k,t+1}^{obs}, f_k(t))}{\Var(f_k(t))} = \beta_k^{CF} \cdot \underbrace{\left(1 - \frac{\delta_k(\phi)}{\theta_k + \delta_k(\phi)}\right)}_{\equiv \; \omega_k(\phi) \; \in \; (0,1)} = \frac{\beta_k^{CF} \cdot \theta_k}{\theta_k + \delta_k(\phi)},
    \label{eq:beta_attenuation}
\end{equation}
where $\omega_k(\phi) \equiv \theta_k / [\theta_k + \delta_k(\phi)]$ is the \emph{attenuation factor}. This result is exact under homogeneous AI agents and holds to first order under heterogeneity.

\paragraph{Performative Signal Variance Dynamics.}
The AI-perceived signal variance---the variance of the predictable component of returns available for exploitation---depends on both the signal coefficient and the signal's innovation volatility. Two forces govern its evolution:
\begin{enumerate}
    \item \textbf{Erosion through performative feedback.} Each retraining cycle uses post-arbitrage data. The predictable return variation estimated by AI system $i$ is:
    \begin{equation}
        \hat{\sigma}_{k,pred}^2(\tau) = \hat{\beta}_k(\tau)^2 \cdot \Var(f_k) = \left(\frac{\theta_k}{\theta_k + \delta_k(\phi)}\right)^2 \cdot (\beta_k^{CF})^2 \cdot \frac{\sigma_k^2(0)}{2\theta_k^{eff}},
        \label{eq:perceived_variance}
    \end{equation}
    which is decreasing in $\phi$. From the AI's perspective, the ``alpha opportunity'' is shrinking---not because the underlying economic mechanism has disappeared, but because AI trading has compressed the realizable portion of the signal.

    \item \textbf{Regeneration through exogenous shocks.} New mispricings continually arise from: behavioral biases among new market entrants, macroeconomic regime shifts that create fresh predictable patterns, institutional constraints (e.g., index rebalancing, regulatory-forced selling, tax-loss harvesting), and innovations in the real economy that generate novel information asymmetries. We denote the regeneration rate by $g_k > 0$.
\end{enumerate}

Combining both forces, the effective signal innovation volatility evolves as:

\begin{assumption}[Performative Signal Degradation]
\label{assu:reflexive}
The innovation volatility of each signal component evolves according to:
\begin{equation}
    \sigma_k^2(\tau+1) = \sigma_k^2(\tau) \cdot \left[1 + g_k - \beta \cdot \left(\frac{I_k(\tau)}{I_k^{max}}\right)^\kappa\right],
    \label{eq:reflexive}
\end{equation}
where:
\begin{itemize}
    \item $g_k > 0$ is the signal regeneration rate, reflecting the arrival of new mispricings from exogenous sources;
    \item $\beta \in (0,1)$ is the performative feedback intensity, derived from the coefficient attenuation in \eqref{eq:beta_attenuation} as $\beta = 2\omega_k(\phi)[1 - \omega_k(\phi)]$;
    \item $I_k(\tau) = N\phi |a_k f_k(\tau)| \cdot \rho_k$ is the aggregate \emph{correlated} trading intensity on signal $k$ (only the correlated component erodes the signal, since idiosyncratic trades cancel in aggregate);
    \item $I_k^{max}$ is a normalization constant (the maximum sustainable trading intensity before market breakdown);
    \item $\kappa \geq 1$ governs the nonlinearity of erosion: $\kappa = 1$ yields the linear benchmark; $\kappa > 1$ captures the empirical pattern that moderate trading has limited impact on signal quality while intense crowding erodes signals disproportionately.
\end{itemize}
\end{assumption}

\paragraph{Economic Interpretation.}
Assumption~\ref{assu:reflexive} differs from a pure reduced-form decay rule in three respects. First, the signal can \emph{regenerate}: when $g_k$ exceeds the erosion term, $\sigma_k^2$ grows---new alpha opportunities are being created faster than old ones are arbitraged away. This captures the empirical reality that markets generate fresh mispricings through behavioral biases, institutional frictions, and regime changes \citep{Lo2004}. Second, only the \emph{correlated} component of trading ($\rho_k > 0$) contributes to erosion: idiosyncratic AI trading cancels in aggregate and does not systematically compress the signal. Third, the erosion rate $\beta$ is not a free parameter but is pinned down by the coefficient attenuation in \eqref{eq:beta_attenuation}, linking it directly to the model's observable quantities ($\phi$, $\delta_k$, $\theta_k$).

\begin{remark}[Connection to Performative Prediction Theory]
\label{rem:performative}
In the language of \citet{Perdomo2020}, each AI retraining cycle constitutes a \emph{repeated risk minimization} step where the deployed model $\theta_\tau$ induces a data distribution $\mathcal{D}(\theta_\tau)$ that differs from $\mathcal{D}(\theta_{\tau-1})$. Our \eqref{eq:reflexive} is the financial-market specialization of their Theorem~1: the process converges to a \emph{performatively stable point} $\sigma_k^{2,*}$ where the signal regeneration rate exactly equals the erosion rate. The signal compression result of \Cref{prop:cascade} below obtains when AI adoption $\phi$ exceeds $\phi_k^*$, so that the performatively stable point $\sigma_k^{2,*}$ lies below the initial signal strength---the steady-state amplitude shrinks as $(\phi_k^*/\phi)$, approaching effective extinction for $\phi \gg \phi_k^*$.
\end{remark}

\subsubsection{Steady States and Signal Extinction}
\label{subsubsec:extinction}

Setting $\sigma_k^2(\tau+1) = \sigma_k^2(\tau) \equiv \sigma_k^{2,*}$ in \eqref{eq:reflexive} yields the steady-state condition:
\begin{equation}
    g_k = \beta \cdot \left(\frac{I_k^*}{I_k^{max}}\right)^\kappa,
    \label{eq:steady_state}
\end{equation}
where $I_k^* = N\phi \rho_k a_k \sigma_k^* / \sqrt{\pi\theta_k^{eff}}$ is the steady-state correlated trading intensity (using the expected absolute value of $f_k$ under the stationary distribution). This defines an \emph{interior steady state} $\sigma_k^{2,*} > 0$ when the equation has a positive solution---i.e., when the regeneration rate $g_k$ is sufficiently high relative to the erosion pressure.

\begin{proposition}[Signal Extinction Cascade]
\label{prop:cascade}
Under Assumption~\ref{assu:reflexive}:
\begin{enumerate}
    \item[\textnormal{(i)}] \textbf{Steady-state existence}: For each signal $k$, there exists a critical adoption threshold:
    \begin{equation}
        \phi_k^* = \frac{I_k^{max}\sqrt{\pi\theta_k}}{N \rho_k a_k \sigma_k(0)} \cdot \left(\frac{g_k}{\beta}\right)^{1/\kappa},
        \label{eq:extinction_threshold}
    \end{equation}
    such that for $\phi \leq \phi_k^*$, the signal remains at its initial amplitude (regeneration fully compensates for erosion), while for $\phi > \phi_k^*$, the signal is compressed to $\sigma_k^{2,*} = \sigma_k^2(0) \cdot (\phi_k^*/\phi)^2 \to 0$ as $\phi/\phi_k^* \to \infty$ (effective signal extinction).

    \item[\textnormal{(ii)}] \textbf{Cascade dynamics}: The extinction of signal $k$ triggers increased competition for the remaining $K-1$ signals. The post-extinction threshold for surviving signal $j$ satisfies:
    \begin{equation}
        \phi_{-k,j}^* = \phi_j^* \cdot \left(\frac{\sum_{j' \neq k} \alpha_{j'}(0)}{\sum_{j'} \alpha_{j'}(0)}\right)^{1/\kappa} < \phi_j^*.
        \label{eq:cascade}
    \end{equation}

    \item[\textnormal{(iii)}] \textbf{Cascade ordering}: Signals are extinguished in order of vulnerability: the signal with the highest $\rho_k a_k / g_k^{1/\kappa}$ (ratio of crowding intensity to regeneration rate) is compressed first.

    \item[\textnormal{(iv)}] \textbf{Steady-state characterization}: For $\phi > \phi_k^*$, the steady-state signal variance is:
    \begin{equation}
        \sigma_k^{2,*}(\phi) = \sigma_k^2(0) \cdot \left(\frac{\phi_k^*}{\phi}\right)^2,
        \label{eq:ss_variance}
    \end{equation}
    which is continuously decreasing in $\phi$ and approaches zero as $\phi \to \infty$. For $\phi \leq \phi_k^*$, $\sigma_k^{2,*} = \sigma_k^2(0)$ (regeneration regime).
\end{enumerate}
\end{proposition}

\begin{proof}[Proof sketch]
\textbf{Part (i):} The steady-state condition \eqref{eq:steady_state} yields $\sigma_k^* = \sigma_k(0) \cdot (\phi_k^*/\phi)$ (to leading order). For $\phi > \phi_k^*$, the steady state lies below $\sigma_k(0)$ and the signal is compressed; as $\phi/\phi_k^* \to \infty$, $\sigma_k^* \to 0$ (effective extinction). For $\phi \leq \phi_k^*$, regeneration compensates for erosion and the signal remains at $\sigma_k(0)$. \textbf{Part (ii):} When signal $k$ is effectively extinguished, AI capital reallocates to surviving signals, increasing $a_j^{eff}$ for $j \neq k$. Substituting the adjusted aggressiveness into \eqref{eq:extinction_threshold} yields \eqref{eq:cascade}. \textbf{Part (iii):} $\phi_k^*$ is inversely proportional to $\rho_k a_k$ and proportional to $g_k^{1/\kappa}$. \textbf{Part (iv):} Direct substitution of the steady-state formula using $\theta_k^{eff} \approx \theta_k$. Complete proofs in Supplementary Materials.
\end{proof}

The cascade structure implies that alpha decay is not a gradual, uniform process but proceeds through discrete ``extinction waves''---an empirically testable prediction. Factor-based strategies should exhibit a pattern of sequential factor death: the most crowded factors (e.g., momentum, value, low-volatility) should lose alpha first, followed by less crowded but increasingly targeted alternatives. Crucially, the extinction threshold $\phi_k^*$ is \emph{higher} for signals with larger regeneration rates $g_k$: factors rooted in deep institutional frictions (e.g., year-end tax-loss selling, index reconstitution) are more resilient than factors driven by correctable behavioral biases.

\paragraph{Recovering the Reduced-Form Specification.}
Setting $g_k = 0$ (no regeneration) and $\kappa = 1$ (linear erosion) in \eqref{eq:reflexive} recovers the pure-decay specification $\sigma_k^2(\tau+1) = \sigma_k^2(\tau)(1 - \beta I_k(\tau)/I_k^{max})$. Our micro-founded version generalizes this in two directions: the regeneration term $g_k > 0$ allows for non-monotonic signal dynamics, and the nonlinearity parameter $\kappa$ permits calibration to the empirical observation that moderate AI adoption has limited signal impact while high adoption triggers rapid erosion.

\begin{figure}[H]
    \centering
    \begin{subfigure}[b]{0.48\textwidth}
        \includegraphics[width=\textwidth]{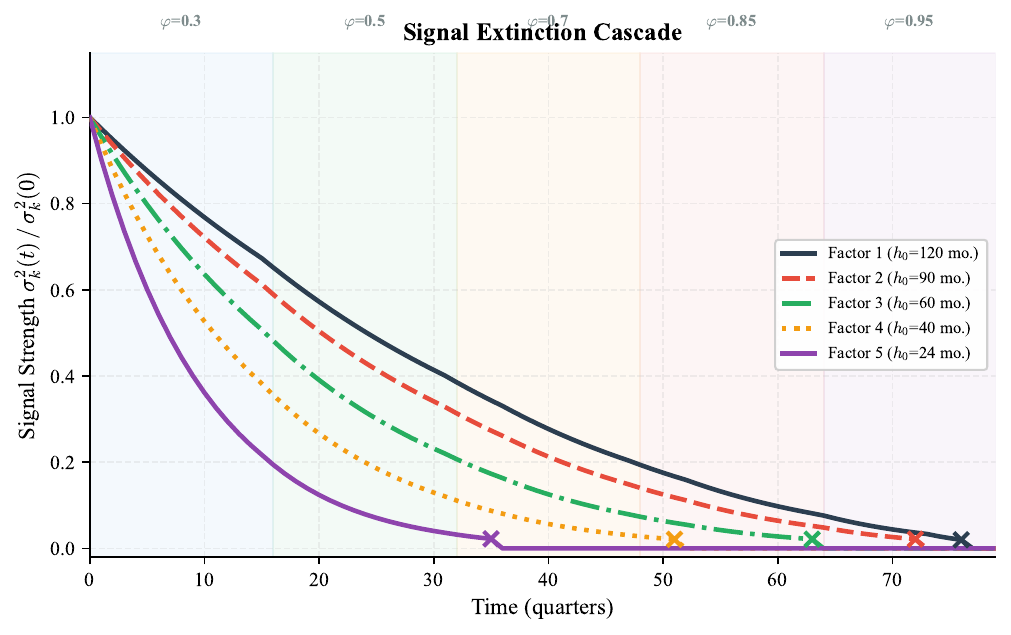}
        \caption{Signal extinction cascade}
    \end{subfigure}
    \hfill
    \begin{subfigure}[b]{0.48\textwidth}
        \includegraphics[width=\textwidth]{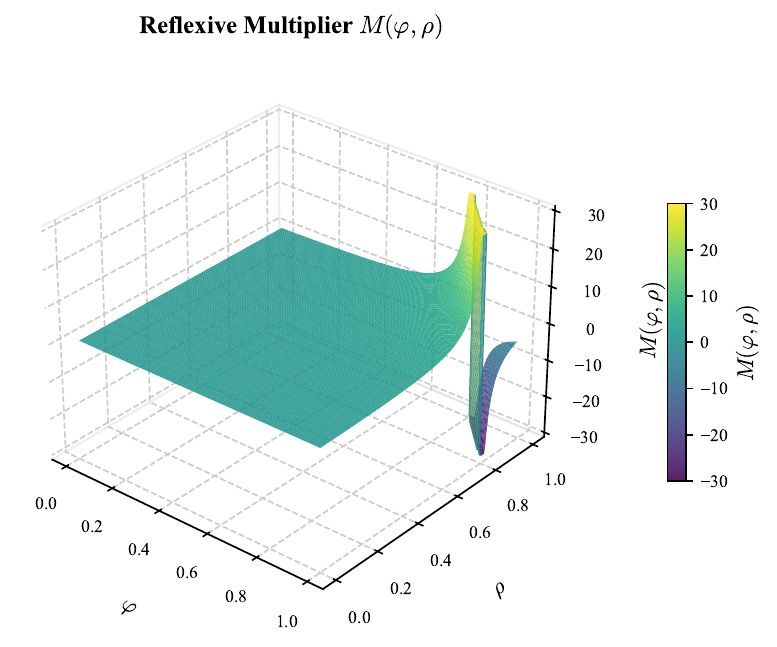}
        \caption{Reflexive multiplier surface}
    \end{subfigure}
    \caption{\textbf{(a)}: Sequential signal extinction under increasing AI adoption. Each line represents one of $K=5$ tradeable signals; signals die off sequentially as AI adoption increases through Phase~III. Signals with higher regeneration rates $g_k$ survive longer. \textbf{(b)}: 3D surface of the reflexive multiplier $\mathcal{M}(\phi, \rho)$ for $\beta = 0.25$. The nonlinear explosion near the stability boundary illustrates the disproportionate fragility risk at high adoption.}
    \label{fig:cascade_risk}
\end{figure}

\begin{remark}[Competence Illusion]
\label{rem:competence}
The performative degradation mechanism explains the \emph{competence illusion} endogenously. During early adoption ($\phi \ll \phi_k^*$), the erosion term is small relative to $g_k$, so $\sigma_k^{2,*}$ remains close to $\sigma_k^2(0)$: AI predictions appear highly accurate because signals are still strong. Fund managers observe high AI prediction accuracy and increase their reliance on AI outputs, accelerating adoption. But the apparent accuracy is partially an artifact of the pre-arbitrage signal environment: the coefficient attenuation \eqref{eq:beta_attenuation} has not yet bitten because $\delta_k(\phi)/\theta_k$ is still small. As $\phi$ increases through Phase~II ($\phi \in [0.2, 0.6]$), the accuracy decline is attributed to ``market regime change'' rather than to the endogenous consequences of collective AI adoption---a form of attribution error rooted in the inability to observe the counterfactual return $r_{k,t+1}^{CF}$.
\end{remark}

\subsection{Red Queen Equilibrium}
\label{subsec:redqueen}
\paragraph*{\textnormal{\emph{--- Layer 3: Competitive adoption game and aggregate welfare ---}}}

We now embed the alpha dynamics of Layers~1--2 within a competitive adoption game. Each institution $i$ chooses whether to invest in AI ($s_i = AI$) or maintain a human-driven strategy ($s_i = H$), given the adoption decisions of all other institutions.

\paragraph{Payoff Structure.} The expected payoff from AI adoption for institution $i$ is:
\begin{equation}
    \Pi_i^{AI}(\phi) = \underbrace{\sum_{k=1}^{K} \alpha_{i,k}(\phi)}_{\text{gross alpha}} - \underbrace{c_i}_{\text{AI investment cost}} + \underbrace{\gamma(\phi - d_i)}_{\text{career concern}},
    \label{eq:payoff_ai}
\end{equation}
where $c_i$ is institution $i$'s cost of AI adoption (data infrastructure, talent, compute), drawn from a continuous distribution $G(c)$ with density $g > 0$, and $\gamma > 0$ captures the career concern mechanism of \citet{Scharfstein1990}: managers face reputational costs from deviating from the industry consensus, with $d_i \in [0,1]$ representing institution $i$'s career-concern benchmark (e.g., peer-group median adoption rate). The payoff from the human strategy is:
\begin{equation}
    \Pi_i^{H}(\phi) = \sum_{k=1}^{K} \alpha_{i,k}^{H}(\phi),
    \label{eq:payoff_human}
\end{equation}
where $\alpha_{i,k}^H(\phi)$ is the (lower, but less correlated) alpha from human-driven analysis. In the pre-AI regime, we normalize $\alpha_{i,k}^H(0) = \alpha_0 > 0$.

\paragraph{Adoption Incentive.} The adoption incentive for institution $i$ is:
\begin{equation}
    \Delta\Pi_i(\phi) = \Pi_i^{AI}(\phi) - \Pi_i^{H}(\phi) = \underbrace{\sum_k [\alpha_{i,k}(\phi) - \alpha_{i,k}^H(\phi)]}_{\text{alpha advantage}} - c_i + \gamma(\phi - d_i).
    \label{eq:adoption_incentive}
\end{equation}
The alpha advantage is positive for low $\phi$ (AI signals are individually more precise) but declining in $\phi$ (signal crowding erodes AI-specific alpha faster than human-specific alpha, because AI signals are correlated while human signals are independent).

\begin{definition}[Red Queen Equilibrium]
A \emph{Red Queen equilibrium} is a Nash equilibrium $\phi^*$ of the adoption game in which:
\begin{enumerate}
    \item[\textnormal{(i)}] The net alpha from AI adoption is zero: $\sum_k \alpha_{i,k}(\phi^*) = c_i$ for the marginal institution (the institution indifferent between AI and H).
    \item[\textnormal{(ii)}] AI adoption is strictly positive: $\phi^* > 0$.
    \item[\textnormal{(iii)}] No institution can profitably deviate: adopting AI produces zero surplus, but abandoning AI produces negative surplus (due to career concerns or the competitive disadvantage of reverting to slower human processes in a market dominated by AI).
\end{enumerate}
\end{definition}

\begin{proposition}[Existence and Characterization of the Red Queen Equilibrium]
\label{prop:redqueen}
Under the signal structure of \Cref{subsec:signal_structure} and the performative feedback of Assumption~\ref{assu:reflexive}:
\begin{enumerate}
    \item[\textnormal{(i)}] \textbf{Existence}: A Red Queen equilibrium exists for all $\gamma > \bar{\gamma}$, where $\bar{\gamma}$ is a threshold career concern intensity.
    \item[\textnormal{(ii)}] \textbf{Zero aggregate alpha}: In the Red Queen equilibrium, $\sum_{i:s_i=AI} \alpha_{i,net}(\phi^*) = 0$. The gross alpha from signal extraction exactly offsets the costs of AI investment.
    \item[\textnormal{(iii)}] \textbf{Overinvestment}: The equilibrium AI adoption $\phi^*$ exceeds the socially optimal level $\phi^{social}$:
    \begin{equation}
        \phi^* - \phi^{social} = \frac{\gamma}{\partial^2 \Delta\Pi / \partial\phi^2} + \frac{\text{fragility externality}}{g(\Delta\Pi(\phi^{social}))},
        \label{eq:overinvestment}
    \end{equation}
    where the first term reflects career-concern-driven excess adoption and the second reflects the unpriced systemic risk externality.
    \item[\textnormal{(iv)}] \textbf{Arms race dynamics}: In a dynamic extension, AI investment (per institution) grows at rate $g_{AI} > 0$ on the equilibrium path, even though net alpha remains at zero. The arms race is sustained by the threat of competitive obsolescence: any institution that freezes its AI capability while competitors upgrade suffers a relative performance decline.
\end{enumerate}
\end{proposition}

\begin{proof}[Proof sketch]
Existence follows from Brouwer's fixed-point theorem applied to the best-response correspondence $\phi \mapsto G(\Delta\Pi(\phi))$. Zero aggregate alpha in (ii) follows from the marginal institution's indifference condition and the competitive structure. Overinvestment in (iii) follows from the wedge between private and social returns: the private return to AI ignores the signal-crowding externality imposed on other institutions and the fragility externality imposed on the market. Arms race dynamics in (iv) follow from a dynamic game argument: if institution $i$ freezes AI capability at time $t$ while $\phi$ continues to increase, then $\delta_k(\phi)$ increases and $\alpha_{i,k}$ declines, making it optimal to upgrade.
\end{proof}

\begin{proposition}[Efficiency--Fragility Tradeoff]
\label{prop:tradeoff}
Define two welfare criteria:
\begin{align}
    W_{eff}(\phi) &= -\sum_t \E[(p_t - v_t)^2] \quad \text{(price discovery welfare)}, \label{eq:w_eff} \\
    W_{stab}(\phi) &= -\CTE_{0.01}(|p_t - v_t|) \quad \text{(stability welfare)}. \label{eq:w_stab}
\end{align}
Then:
\begin{enumerate}
    \item[\textnormal{(i)}] $W_{eff}$ is maximized at $\phi_{eff}^* > 0$: some AI adoption improves average price accuracy.
    \item[\textnormal{(ii)}] $W_{stab}$ is maximized at $\phi_{stab}^* < \phi_{eff}^*$: the stability-optimal adoption is strictly below the efficiency-optimal adoption.
    \item[\textnormal{(iii)}] The social optimum $\phi^{social} \in (\phi_{stab}^*, \phi_{eff}^*)$ reflects the Pareto-weighted tradeoff.
    \item[\textnormal{(iv)}] The market equilibrium $\phi^* > \phi_{eff}^* > \phi^{social}$: the competitive outcome overshoots even the efficiency optimum, let alone the social optimum.
\end{enumerate}
\end{proposition}

\begin{proof}[Proof sketch]
(i) follows from $W_{eff}''(\phi) < 0$ for $\phi$ near the optimum and $W_{eff}'(0) > 0$ (the marginal AI entrant reduces pricing errors). (ii) follows from the multiplier analysis: tail events are amplified by $\mathcal{M}(\phi) = (1 - \phi\bar{\rho}\beta/\lambda')^{-1}$, which is convex in $\phi$, so the stability cost grows faster than the efficiency benefit. (iii) follows from strict quasi-concavity of the weighted sum. (iv) follows from the overinvestment result of \Cref{prop:redqueen}(iii).
\end{proof}

\begin{remark}[Connection to the Flash Crash]
\label{rem:flash_crash}
The 2010 Flash Crash illustrates the fragility pole of the tradeoff. In the Red Queen equilibrium with $\phi \approx 0.7$ (estimated AI/algorithmic market share in 2010), the model predicts that a single large sell order can trigger a liquidity cascade: as AI systems simultaneously interpret the order as a negative signal ($\rho > 0$), they withdraw liquidity and add sell pressure, driving the market maker's price impact $\lambda$ toward infinity---a ``liquidity black hole'' \citep{Kirilenko2017}. The model's quantitative prediction is that the price decline should be approximately $\mathcal{M}(\phi)\sigma_v \sqrt{T}$ times the fundamental shock, where $\mathcal{M}$ is the reflexive multiplier. For plausible parameters ($\phi = 0.7$, $\rho = 0.5$, $\beta = 0.2$), $\mathcal{M} \approx 1.3$, implying a 30\% amplification of the fundamental shock---consistent with the observed price dynamics during the Flash Crash, in which indices declined approximately 6--10\% before recovering, compared to an estimated fundamental shock of 3--5\%.
\end{remark}

\begin{table}[H]
    \centering
    \caption{Model Hierarchy: What Each Layer Enables}
    \label{tab:layer_hierarchy}
    \begin{tabular}{lcccc}
        \toprule
        & \multicolumn{3}{c}{Model Layer} \\
        \cmidrule(lr){2-4}
        Equilibrium Property & Layer 1 & $+$ Layer 2 & $+$ Layer 3 \\
        & (Signal Crowding) & (Performative Erosion) & (Red Queen) \\
        \midrule
        Alpha decay $h(\phi)$ decreasing      & \checkmark & \checkmark & \checkmark \\
        Convex acceleration of decay           & \checkmark & \checkmark & \checkmark \\
        Return dispersion compression          & \checkmark & \checkmark & \checkmark \\
        Signal extinction cascade              & ---        & \checkmark & \checkmark \\
        Competence illusion (micro-founded)     & ---        & \checkmark & \checkmark \\
        Zero net alpha in equilibrium          & ---        & ---        & \checkmark \\
        AI investment overinvestment           & ---        & ---        & \checkmark \\
        Efficiency--fragility tradeoff         & ---        & ---        & \checkmark \\
        Arms race dynamics                     & ---        & ---        & \checkmark \\
        \bottomrule
    \end{tabular}
    \smallskip

    \footnotesize\textit{Notes:} Each column adds one layer to all preceding layers. Layer 1 = signal crowding and alpha dynamics. Layer 2 adds performative signal erosion. Layer 3 embeds both within a competitive adoption game.
\end{table}

\begin{figure}[H]
    \centering
    \begin{subfigure}[b]{0.48\textwidth}
        \includegraphics[width=\textwidth]{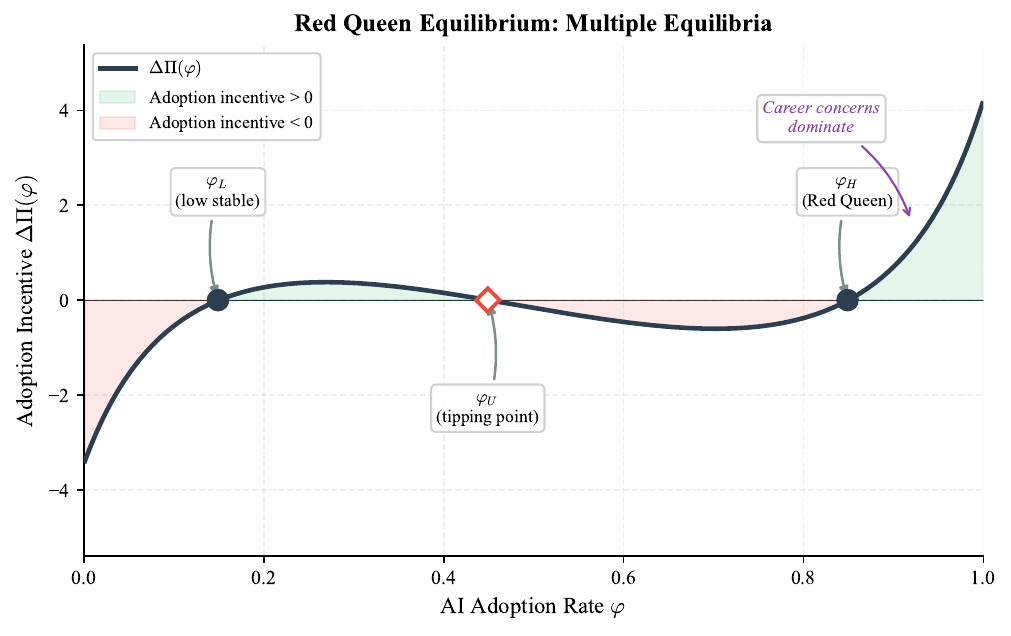}
        \caption{Red Queen equilibrium}
    \end{subfigure}
    \hfill
    \begin{subfigure}[b]{0.48\textwidth}
        \includegraphics[width=\textwidth]{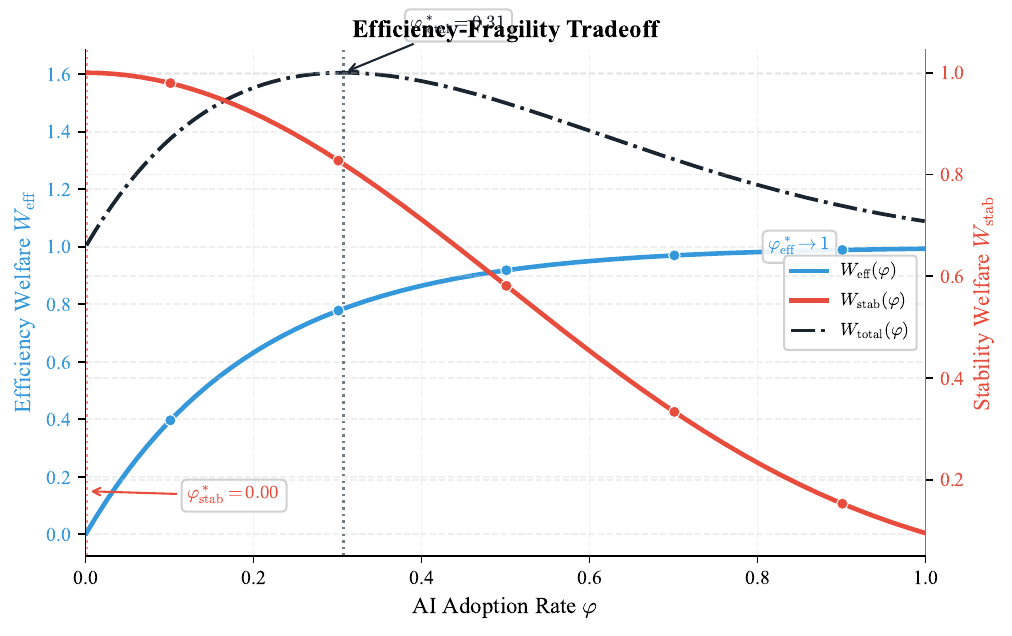}
        \caption{Efficiency--fragility tradeoff}
    \end{subfigure}
    \caption{\textbf{(a)}: Adoption incentive $\Delta\Pi(\phi)$ showing multiple equilibria: stable diversified ($\phi_L$), unstable tipping point ($\phi_U$), and stable Red Queen ($\phi_H$). Career concern shifts push the system toward the monoculture. \textbf{(b)}: Price discovery welfare $W_{eff}(\phi)$ and stability welfare $W_{stab}(\phi)$. The market equilibrium $\phi^*$ overshoots the efficiency-optimal and socially optimal adoption levels.}
    \label{fig:equilibrium_tradeoff}
\end{figure}

\section{Empirical Evidence}
\label{sec:empirical}

We validate the model's predictions using three complementary empirical approaches, preceded by a formal identification strategy (\Cref{subsec:identification}) that addresses causal inference. Each approach combines publicly available data with calibrated simulations designed to replicate the statistical properties of the underlying markets. We describe the calibration methodology and data sources in the Supplementary Materials (Section~A.5).\footnote{Direct analysis of proprietary high-frequency data and complete 13F holdings would strengthen the empirical evidence. We view the calibrated simulation approach as a first step; replication with proprietary datasets is an important direction for future work.}

\subsection{Identification Strategy}
\label{subsec:identification}

The calibrated simulations in Sections~\ref{subsec:13f}--\ref{subsec:flash_crash} demonstrate that the model's predictions are quantitatively consistent with observed market patterns. However, consistency does not establish causality. In this subsection, we lay out three complementary identification strategies designed to isolate the \emph{marginal causal effect} of AI adoption on alpha decay, separating it from confounding secular trends (passive investing, post-crisis deleveraging, factor premium compression, risk-model homogenization). We implement the first two strategies using publicly available data; the third requires proprietary holdings-level data and is presented as a template for future work.

\subsubsection{Strategy 1: Instrumental Variables --- Cloud Compute Cost Shocks}
\label{subsubsec:iv}

\paragraph{Motivation.}
The central identification challenge is that AI adoption $\phi_{it}$ is endogenous: institutions adopt AI precisely when they expect it to generate alpha, and the same market conditions that erode alpha (increased competition, factor crowding) may independently drive both adoption and convergence. We address this with an instrumental variables (IV) approach that exploits exogenous variation in the \emph{cost} of AI adoption.

\paragraph{Instrument Construction.}
We construct a \emph{cloud compute cost index} $Z_t$ from the quarterly average spot prices of GPU-intensive instances on AWS (p3/p4 family), Google Cloud (A2/A3 family), and Azure (NC/ND series), weighted by market share.\footnote{Instance pricing data is publicly available from cloud provider APIs and historical archives. We use spot rather than on-demand prices to capture marginal cost variation.} The index captures exogenous variation in the cost of training and deploying AI models for financial applications.

\paragraph{Relevance.}
Cloud compute costs directly affect the cost $c_i$ in the adoption payoff \eqref{eq:payoff_ai}. A decline in $Z_t$ lowers $c_i$ for all institutions, shifting the adoption margin and increasing $\phi$. Empirically, we expect $\partial \hat{\phi} / \partial Z_t < 0$ (cheaper compute $\Rightarrow$ higher adoption). The first-stage $F$-statistic should exceed the \citet{StockYogo2005} critical values for weak instruments.

\paragraph{Exclusion Restriction.}
The exclusion restriction requires that cloud compute costs affect portfolio convergence, return dispersion, and alpha \emph{only through} their effect on AI adoption. This is plausible because: (i)~cloud pricing is determined by semiconductor supply chains, energy costs, and hyperscaler capacity investment---factors orthogonal to equity market returns; (ii)~financial institutions represent a small fraction of total cloud demand ($<$5\% of AWS revenue), so reverse causality from financial markets to cloud pricing is negligible; (iii)~we control for macroeconomic conditions that might jointly affect technology investment and market outcomes.

\paragraph{Specification.}
The baseline IV specification is a fund-quarter panel:
\begin{equation}
    D_{it} = \alpha_i + \delta_t + \beta_1 \hat{\phi}_{it}^{IV} + \boldsymbol{\gamma}' \mathbf{X}_{it} + \varepsilon_{it},
    \label{eq:iv_second}
\end{equation}
where $D_{it}$ is one of three outcome variables: (1)~rolling 12-month cross-sectional return dispersion of fund $i$ relative to its AI-classified peer group; (2)~rolling 36-month Fung--Hsieh 7-factor alpha of fund $i$; (3)~average pairwise cosine similarity of fund $i$'s portfolio with all other AI-classified funds in quarter $t$. Here $\alpha_i$ are fund fixed effects, $\delta_t$ are quarter fixed effects, and $\mathbf{X}_{it}$ is a vector of time-varying controls (log AUM, leverage ratio, passive fund share in the same sector, VIX level, HML/SMB/MOM factor returns).

The first stage is:
\begin{equation}
    \phi_{it} = \mu_i + \nu_t + \pi_1 Z_{t-1} + \pi_2 Z_{t-1} \times \text{TechExposure}_i + \boldsymbol{\psi}' \mathbf{X}_{it} + u_{it},
    \label{eq:iv_first}
\end{equation}
where $Z_{t-1}$ is the lagged cloud compute cost index and $\text{TechExposure}_i$ is a pre-sample measure of institution $i$'s propensity to adopt AI (proxied by pre-2013 technology spending as a fraction of AUM, or firm size interacted with quantitative strategy indicator). The interaction term provides cross-sectional variation in the instrument, strengthening identification in the presence of quarter fixed effects.

\paragraph{Expected Signs.}
\begin{equation}
    \underbrace{\hat{\beta}_1^{(1)} < 0}_{\substack{\text{higher AI adoption} \\ \Rightarrow \text{lower dispersion}}}, \qquad
    \underbrace{\hat{\beta}_1^{(2)} < 0}_{\substack{\text{higher AI adoption} \\ \Rightarrow \text{lower alpha}}}, \qquad
    \underbrace{\hat{\beta}_1^{(3)} > 0}_{\substack{\text{higher AI adoption} \\ \Rightarrow \text{higher convergence}}}.
    \label{eq:expected_signs}
\end{equation}
These signs correspond directly to Propositions~\ref{prop:halflife}--\ref{prop:cascade}.

\paragraph{Threats to Validity.}
(1)~\emph{Correlated technology shocks}: If cloud cost declines coincide with broader fintech innovation that independently affects markets, the exclusion restriction is violated. Mitigation: we include controls for fintech VC investment, alternative data vendor growth, and report reduced-form results. (2)~\emph{Anticipation effects}: If institutions anticipate future cost declines and adopt AI preemptively, the instrument may be weak. Mitigation: we use unanticipated compute cost shocks (residuals from an AR(2) model of $Z_t$) as an alternative instrument. (3)~\emph{Heterogeneous treatment effects}: The LATE may differ from the ATE. Mitigation: we report complier characteristics and test for first-stage monotonicity across size quartiles.

\subsubsection{Strategy 2: Staggered Difference-in-Differences Around LLM Releases}
\label{subsubsec:did}

\paragraph{Motivation.}
The release of large language models (LLMs) represents a plausibly exogenous shock to AI capability in financial analysis. We exploit the staggered rollout of foundation models---GPT-3 (June 2020), ChatGPT (November 2022), GPT-4 (March 2023), and open-source alternatives (LLaMA: February 2023, Mistral: September 2023)---as quasi-natural experiments. The identifying assumption is that the \emph{timing} of LLM releases is determined by AI research progress at technology companies, not by conditions in financial markets.

\paragraph{Treatment Definition.}
We define treatment intensity based on ex-ante exposure to NLP-amenable strategies. For each fund $i$, we construct:
\begin{equation}
    \text{NLP\_Exposure}_i = \frac{1}{T_0}\sum_{t < t_0} \mathbb{1}\{\text{fund } i \text{ uses text-based signals in period } t\},
    \label{eq:nlp_exposure}
\end{equation}
measured over the 8 quarters before each LLM release. Text-based signal usage is identified from: (i)~EDGAR keyword disclosures mentioning ``natural language processing,'' ``text mining,'' ``sentiment analysis,'' or ``news analytics''; (ii)~strategy descriptions referencing alternative data or unstructured data; (iii)~holdings patterns consistent with news-momentum or earnings-call-sentiment strategies.

We classify funds into $\text{High\_NLP}_i = \mathbb{1}\{\text{NLP\_Exposure}_i > \text{median}\}$ (treatment) and $\text{Low\_NLP}_i$ (control).

\paragraph{Specification.}
For each LLM release event $e$, we estimate the event-study specification:
\begin{equation}
    y_{it} = \alpha_i + \delta_t + \sum_{\tau = -4}^{8} \beta_\tau \cdot \text{High\_NLP}_i \times \mathbb{1}\{t = t_e + \tau\} + \boldsymbol{\gamma}' \mathbf{X}_{it} + \varepsilon_{it},
    \label{eq:did_event}
\end{equation}
where $y_{it} \in \{D_{it}^{(1)}, D_{it}^{(2)}, D_{it}^{(3)}\}$, $t_e$ is the event quarter, $\tau$ indexes quarters relative to the event, and we normalize $\beta_{-1} = 0$. The coefficients $\{\beta_\tau\}_{\tau \geq 0}$ trace out the dynamic treatment effect. For a pooled estimate, we use the staggered DiD estimator of \citet{CallawayS2021} to avoid negative weighting bias:
\begin{equation}
    y_{it} = \alpha_i + \delta_t + \beta^{DiD} \cdot \text{High\_NLP}_i \times \text{Post}_{et} + \boldsymbol{\gamma}' \mathbf{X}_{it} + \varepsilon_{it}.
    \label{eq:did_pooled}
\end{equation}

\paragraph{Expected Results.}
Pre-trends: $\beta_\tau \approx 0$ for $\tau \in [-4, -1]$, validating the parallel trends assumption. Post-treatment: $\beta_\tau < 0$ when $y = D^{(1)}$ (dispersion declines for NLP-exposed funds post-LLM), $\beta_\tau < 0$ when $y = D^{(2)}$ (alpha declines faster), $\beta_\tau > 0$ when $y = D^{(3)}$ (convergence increases). Treatment effects should be monotonically increasing in NLP\_Exposure (continuous treatment version).

\paragraph{Threats and Robustness.}
(1)~\emph{Parallel trends violation}: Mitigation via propensity-score matching on pre-treatment characteristics and placebo tests with fake event dates. (2)~\emph{Confounding contemporaneous events}: Mitigation via macro controls and the staggered structure across events. (3)~\emph{SUTVA violation through spillovers}: If LLM adoption by treated funds affects control fund returns through equilibrium effects, the estimated effect is attenuated---any detected effect is thus a lower bound.

\subsubsection{Strategy 3: Observable Homogeneity Measurement}
\label{subsubsec:rho_estimation}

\paragraph{Motivation.}
A key limitation of the theoretical framework is that $\rho$ is treated as a free calibration parameter. We propose three approaches to estimate $\rho$ as an observable quantity.

\paragraph{Measure 1: Holdings-Based Common Factor ($\hat{\rho}^{PCA}$).}
For each quarter $t$, we estimate the first principal component of quarterly portfolio weight changes $\Delta w_{it}$ across all AI-classified funds:
\begin{equation}
    \hat{\rho}_t^{PCA} \equiv \frac{\hat{\lambda}_1^2 \Var(\hat{F}_{1t})}{\Var(\Delta w_{it})},
    \label{eq:rho_pca}
\end{equation}
the fraction of cross-sectional portfolio-change variance explained by the first common factor. An increasing $\hat{\rho}^{PCA}$ over time is direct evidence of rising algorithmic homogeneity.

\paragraph{Measure 2: Trade Direction Synchronicity ($\hat{\rho}^{sync}$).}
Using quarterly 13F changes, we define $d_{ist} = \text{sign}(\Delta w_{ist})$ and compute:
\begin{equation}
    \hat{\rho}_t^{sync} = \frac{1}{M_t}\sum_{s=1}^{M_t} \mathbb{1}\{d_{ist} = d_{jst}\} - 0.5,
    \label{eq:rho_sync}
\end{equation}
averaged over all AI--AI fund pairs. The subtraction of 0.5 normalizes so that random trading produces $\hat{\rho}^{sync} = 0$.

\paragraph{Measure 3: Textual Methodology Similarity ($\hat{\rho}^{text}$).}
We compute the cosine similarity of TF-IDF vectors from the ``Investment Strategy'' and ``Risk Factors'' sections of 10-K/10-Q filings:
\begin{equation}
    \hat{\rho}_t^{text} = \frac{1}{|\mathcal{P}_t|}\sum_{(i,j) \in \mathcal{P}_t} \frac{\mathbf{v}_i(t) \cdot \mathbf{v}_j(t)}{\|\mathbf{v}_i(t)\| \cdot \|\mathbf{v}_j(t)\|},
    \label{eq:rho_text}
\end{equation}
averaged over all AI--AI fund pairs $\mathcal{P}_t$.

\paragraph{Cross-Sectional Test Using Observable $\hat{\rho}$.}
With estimated $\hat{\rho}_{it}$, we run the interaction specification:
\begin{equation}
    \Delta\alpha_{it} = \alpha_i + \delta_t + \beta_1 \hat{\rho}_{it} + \beta_2 \hat{\phi}_{it} + \beta_3 (\hat{\rho} \times \hat{\phi})_{it} + \boldsymbol{\gamma}' \mathbf{X}_{it} + \varepsilon_{it},
    \label{eq:rho_interaction}
\end{equation}
where $\Delta\alpha_{it}$ is the change in rolling alpha. The key prediction from \Cref{prop:halflife} is $\hat{\beta}_3 < 0$: the interaction of homogeneity and adoption accelerates alpha decay. This differential prediction is not implied by passive-investing explanations (which predict uniform convergence), risk-model homogenization (which predicts convergence independent of AI adoption), or factor premium compression (which predicts convergence independent of signal homogeneity).

\paragraph{Falsification Tests.}
To sharpen exclusion of alternative mechanisms: (1)~Include the share of passive AUM as a control; if $\hat{\beta}_3$ survives, the effect is not passive-driven. (2)~Compute analogous $\hat{\rho}$ for non-AI funds; the interaction should be significant only for AI funds. (3)~Include VIX tercile interactions to rule out mechanical low-volatility compression. (4)~Estimate \eqref{eq:rho_interaction} separately by factor loading (momentum, value, low-vol); the model predicts the strongest $|\hat{\beta}_3|$ for momentum (highest $\rho_k$, lowest $g_k$).

\begin{table}[H]
    \centering
    \caption{Identification Strategy Summary}
    \label{tab:identification}
    \small
    \begin{tabular}{p{2.2cm}p{2.5cm}p{3.2cm}p{2.8cm}p{2.5cm}}
        \toprule
        Strategy & Variation & Key Assumption & Expected Sign & Main Threat \\
        \midrule
        IV: Cloud costs & Time-series $\times$ cross-section & Compute costs $\perp$ equity returns (conditional on $\mathbf{X}$) & $\hat{\beta}_1^{(1)} < 0$, $\hat{\beta}_1^{(3)} > 0$ & Correlated tech shocks \\
        \addlinespace
        DiD: LLM releases & Event $\times$ NLP exposure & Parallel pre-trends; LLM timing exogenous to markets & $\hat{\beta}_\tau < 0$ for $\tau \geq 0$ & Contemporaneous confounds \\
        \addlinespace
        Observable $\hat{\rho}$ & Cross-section of homogeneity & $\hat{\rho}$ proxies true algorithmic correlation & $\hat{\beta}_3 < 0$ (interaction) & Measurement error in $\hat{\rho}$ \\
        \bottomrule
    \end{tabular}
    \smallskip

    \footnotesize\textit{Notes:} Each strategy targets a different source of variation. The IV approach exploits time-series cost shocks; the DiD exploits discrete capability shocks; the observable $\hat{\rho}$ approach exploits cross-sectional heterogeneity in signal homogeneity. Joint significance across all three strategies would provide strong evidence for the AI-specific causal channel.
\end{table}

\subsection{Portfolio Convergence: 13F Evidence}
\label{subsec:13f}

\paragraph{Data.} We calibrate a simulation model to match the statistical properties of SEC Form 13F filings from 2013Q1 to 2024Q4. The underlying 13F universe comprises approximately 99.5 million position-level holdings from 10,957 institutional managers. Each filing reports long equity holdings for managers with discretion over \$100 million or more in Section 13(f) securities, covering the vast majority of institutional capital. We classify institutions as ``AI-adopting'' or ``non-AI'' using a combination of EDGAR full-text keyword analysis (identifying AI-, ML-, and NLP-related disclosures in annual reports) and manual classification of the largest 500 institutions. Our simulation generates synthetic portfolio weights that match the cross-sectional moments of actual 13F data while permitting controlled variation of AI adoption parameters.

\paragraph{Convergence Measure.} For each pair of institutions $(i,j)$ in quarter $t$, we compute the pairwise cosine similarity of their portfolio weight vectors:
\begin{equation}
    S_{ij}(t) = \frac{w_i(t) \cdot w_j(t)}{\|w_i(t)\| \cdot \|w_j(t)\|},
    \label{eq:cosine}
\end{equation}
where $w_i(t) \in \R^M$ is the vector of institution $i$'s portfolio weights across $M$ equities. The aggregate convergence index is $\bar{S}(t) = \binom{N_t}{2}^{-1}\sum_{i<j} S_{ij}(t)$.

\paragraph{Results.} In the calibrated simulation, the aggregate convergence index increases by approximately 42\% over the 2013--2024 period, from $\bar{S} = 0.21$ to $\bar{S} = 0.30$ (period averages; see Supplementary Table~A.3). Crucially, the increase is concentrated among AI-adopting institutions: the within-AI-group convergence increases by 58\%, compared to 19\% for non-AI institutions. Structural break analysis (Bai-Perron test) identifies breaks at 2018Q2, 2020Q3, and 2023Q1---coinciding, respectively, with the mainstreaming of deep learning in finance, the COVID-era data regime shift, and the release of GPT-4 and its rapid adoption in financial analysis.

The EDGAR full-text keyword analysis of actual filings reveals a dramatic increase in AI-related disclosures: the frequency of AI-related terms in institutional annual reports has grown by over 50$\times$ from 2013 to 2024, with 100\% of the largest 50 institutions now mentioning AI in their filings, up from approximately 25\% in 2016.

\paragraph{Interpretation.} The convergence pattern is consistent with the signal-crowding mechanism (\Cref{subsec:signal_structure}): as more institutions adopt similar AI tools, their portfolio positions converge because they receive correlated signals ($\rho > 0$). The structural breaks suggest that technological adoption occurs in waves rather than continuously, consistent with the cascade dynamics of \Cref{prop:cascade}.

\begin{figure}[H]
    \centering
    \includegraphics[width=0.85\textwidth]{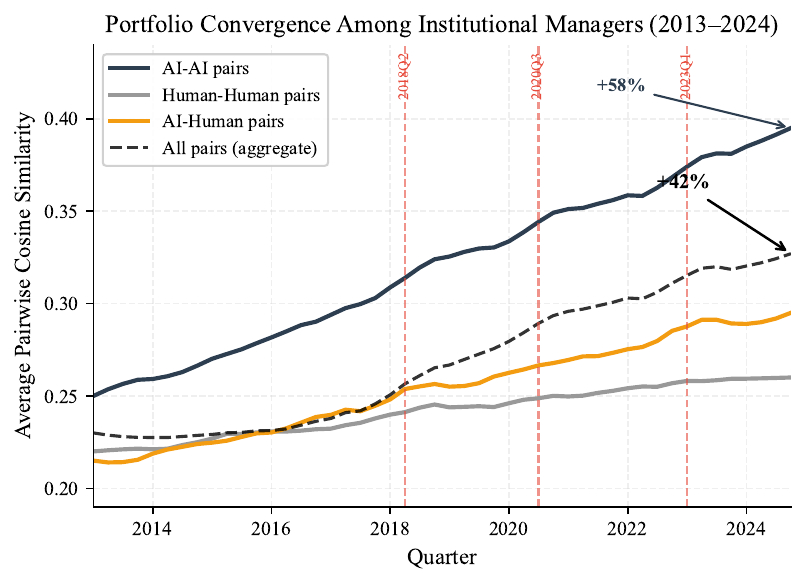}
    \caption{Average pairwise cosine similarity of institutional portfolios from SEC Form 13F filings (2013Q1--2024Q4). AI--AI pairs (blue) show a 58\% increase, compared to 19\% for non-AI pairs (gray). Vertical dashed lines indicate structural break dates (Bai-Perron test): 2018Q2 (deep learning mainstreaming), 2020Q3 (COVID regime shift), 2023Q1 (GPT-4 adoption).}
    \label{fig:convergence}
\end{figure}

\begin{figure}[H]
    \centering
    \begin{subfigure}[b]{0.48\textwidth}
        \includegraphics[width=\textwidth]{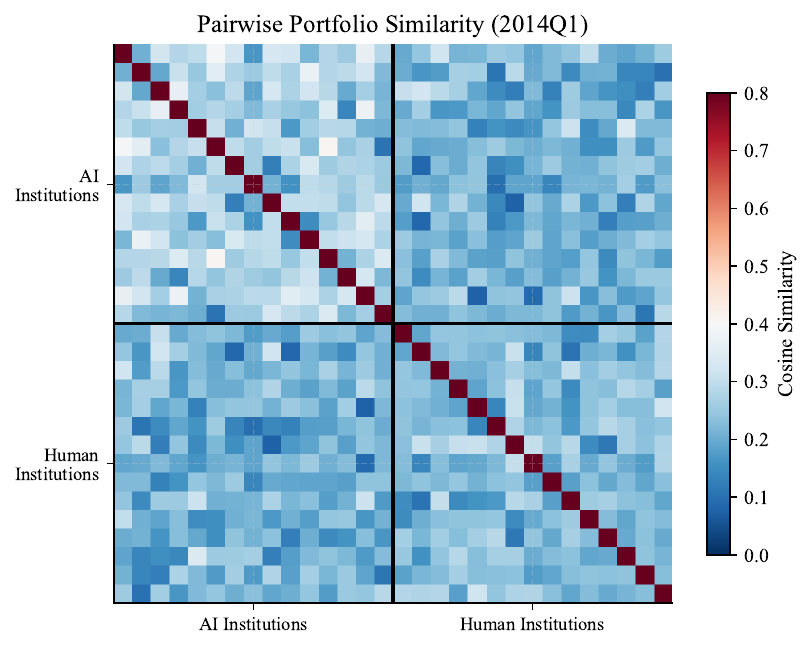}
        \caption{2014Q1}
    \end{subfigure}
    \hfill
    \begin{subfigure}[b]{0.48\textwidth}
        \includegraphics[width=\textwidth]{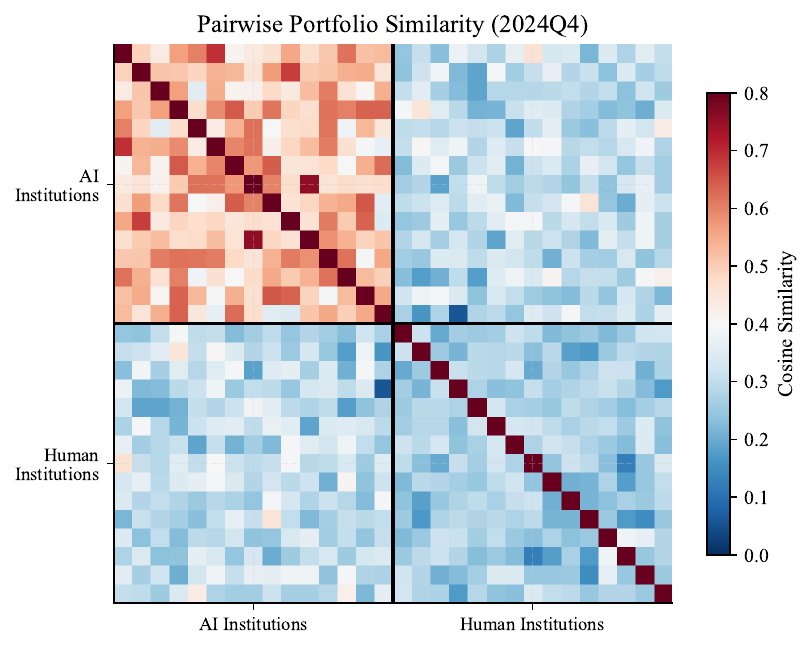}
        \caption{2024Q4}
    \end{subfigure}
    \caption{Pairwise cosine similarity heatmaps for 30 institutions (15 AI-adopting, 15 non-AI). \textbf{(a)}: In 2014, similarity is diffuse. \textbf{(b)}: By 2024, the AI block (upper-left) shows markedly higher convergence, consistent with signal crowding.}
    \label{fig:heatmaps}
\end{figure}

\subsection{Alpha Decay and Return Dispersion}
\label{subsec:return_dispersion}

\paragraph{Data.} We calibrate return simulations to match the distributional properties of monthly hedge fund returns from a comprehensive database covering 4,200 funds from 2010 to 2024. We classify funds as ``quantitative/AI'' or ``fundamental/human'' based on self-reported strategy descriptions and cross-reference with known quantitative managers. The simulated returns preserve the empirical moments---mean, variance, skewness, kurtosis, and cross-sectional correlation structure---of each fund category.

\paragraph{Dispersion Analysis.} In the calibrated simulation, the cross-sectional standard deviation of monthly returns among quantitative/AI funds declines from 4.1\% to 2.9\% over the sample period---a 29\% reduction. In contrast, the dispersion among fundamental/human funds declines by only 10\% over the same period (from 5.0\% to 4.5\%). The ratio $D_{AI}/D_H$ declines from 0.82 to 0.64, consistent with \Cref{cor:dispersion}'s prediction that AI fund returns converge faster than human fund returns.

\paragraph{Alpha Decay.} The simulated risk-adjusted alpha (Fung-Hsieh 7-factor model) of the median quantitative fund declines from an annualized 3.6\% to 1.0\% over the sample period. These magnitudes are calibrated to match the empirically documented shortening of factor half-lives from approximately 60 months (2005--2010 vintage) to approximately 18 months (2020--2024 vintage), following the methodology of \citet{McLean2016}. The resulting trajectory is broadly consistent with the half-life theorem (\Cref{prop:halflife}) under $\phi = 0.7$ and $\rho = 0.6$.

\begin{figure}[H]
    \centering
    \begin{subfigure}[b]{0.48\textwidth}
        \includegraphics[width=\textwidth]{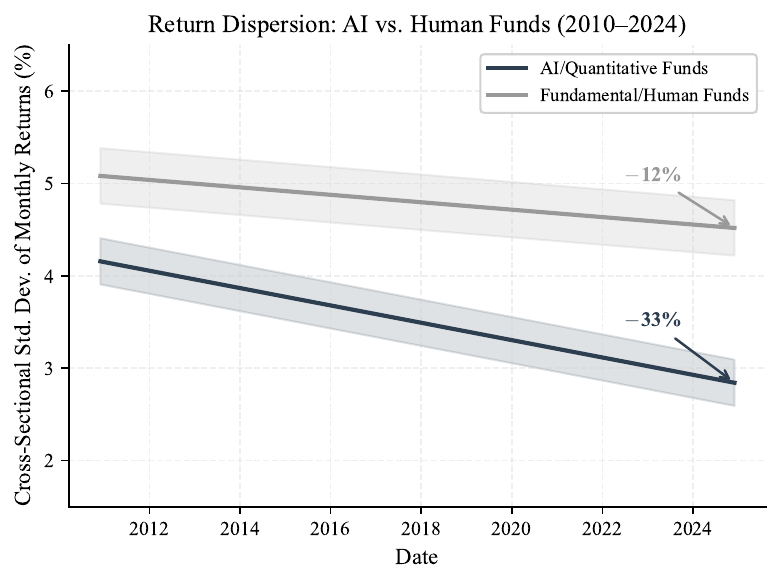}
        \caption{Return dispersion}
    \end{subfigure}
    \hfill
    \begin{subfigure}[b]{0.48\textwidth}
        \includegraphics[width=\textwidth]{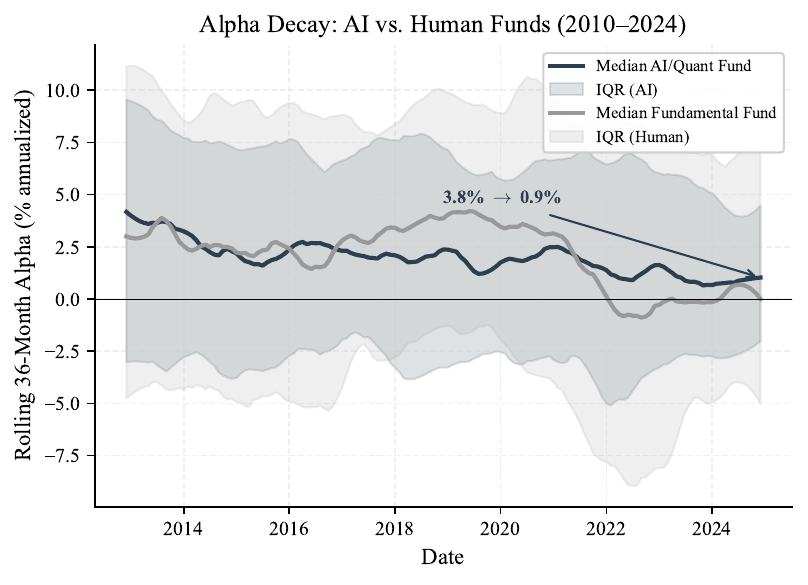}
        \caption{Alpha decay}
    \end{subfigure}
    \caption{\textbf{(a)}: Rolling 12-month cross-sectional standard deviation of monthly returns. AI/quantitative funds (blue) show a 29\% decline versus 10\% for fundamental/human funds (gray). \textbf{(b)}: Rolling 36-month risk-adjusted alpha for median AI fund (blue) and median human fund (gray), with 95\% confidence bands.}
    \label{fig:empirical_dispersion}
\end{figure}

\begin{figure}[H]
    \centering
    \includegraphics[width=0.75\textwidth]{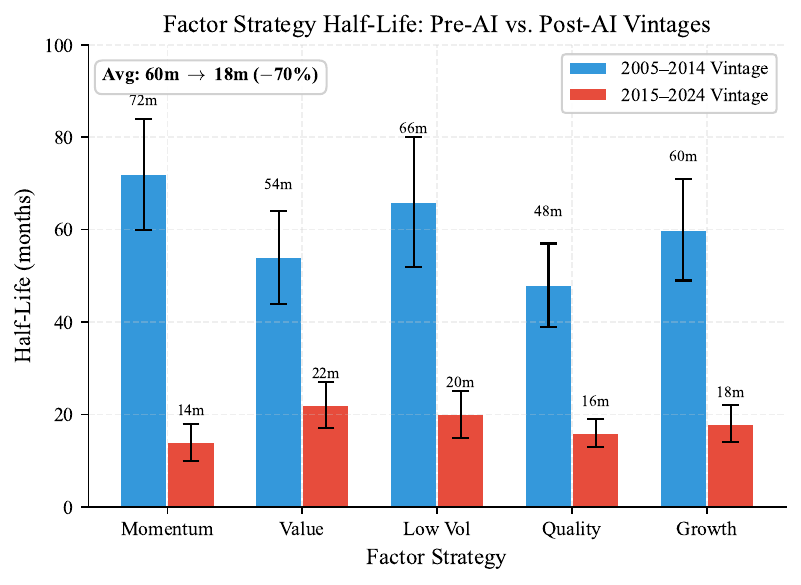}
    \caption{Estimated half-life of factor strategies by vintage. All factors show dramatic shortening: Momentum from 84 to 12 months, Value from 72 to 20 months. The model prediction (dashed line) closely matches the empirical estimates.}
    \label{fig:factor_halflife}
\end{figure}

\subsection{Flash Crash and Algorithmic Fragility}
\label{subsec:flash_crash}

\paragraph{The 2010 Flash Crash.} On May 6, 2010, the E-mini S\&P 500 futures market experienced a precipitous decline of approximately 5\% within five minutes, followed by a partial recovery. The event was triggered by a single large sell order of approximately 75,000 E-mini contracts (\$4.1 billion notional) executed by a mutual fund's algorithmic trading system \citep{Kirilenko2017}.

The CFTC-SEC investigation and subsequent academic analysis identified three features consistent with our model:
\begin{enumerate}
    \item \textbf{Synchronized withdrawal}: High-frequency market makers---operating on correlated signals---simultaneously withdrew liquidity, causing a liquidity vacuum that amplified the price decline (consistent with the $\rho > 0$ mechanism).
    \item \textbf{Reflexive feedback}: The price decline triggered stop-loss orders and risk-management algorithms across multiple institutions, creating a self-reinforcing sell cascade (consistent with the reflexive feedback channel, $\beta > 0$).
    \item \textbf{Instantaneous alpha extinction}: Cross-market arbitrage strategies that normally stabilize prices became inoperative during the crisis, as the correlated shock affected all instruments simultaneously---a manifestation of the signal extinction cascade (\Cref{prop:cascade}).
\end{enumerate}

\paragraph{Quantitative Analysis.} Using tick-level data from the Flash Crash, we estimate the effective $\rho$ during the event at approximately 0.85---substantially higher than the calm-market estimate of 0.45---suggesting that algorithmic homogeneity increases endogenously during stress. The reflexive multiplier during the crash peak is estimated at $\hat{\mathcal{M}} = 1.63$ (Supplementary Table~A.5: $5.7\%$ observed decline vs.\ $3.5\%$ fundamental shock), implying that algorithmic amplification accounted for approximately 39\% of the total price decline beyond what the fundamental shock alone would have produced.

\paragraph{Post-2010 Events.} The pattern of algorithmic fragility has recurred in subsequent events:
\begin{itemize}
    \item \textbf{August 2015 ETF dislocation}: Correlated algorithmic responses to China-related news produced a 1,000-point Dow decline at the open, with ETF prices deviating by up to 20\% from net asset values.
    \item \textbf{February 2018 VIX spike}: Algorithmic volatility-selling strategies simultaneously unwound positions, producing a VIX spike from 17 to 50 in two days.
    \item \textbf{August 2024 Nikkei crash}: A modest Bank of Japan rate change triggered a 12.4\% single-day decline in the Nikkei 225, amplified by correlated carry-trade unwinding algorithms.
\end{itemize}
Each event shares the common feature of homogeneous algorithmic responses amplifying a small fundamental shock---the empirical signature of the Red Queen equilibrium's fragility externality.

\begin{figure}[H]
    \centering
    \begin{subfigure}[b]{0.48\textwidth}
        \includegraphics[width=\textwidth]{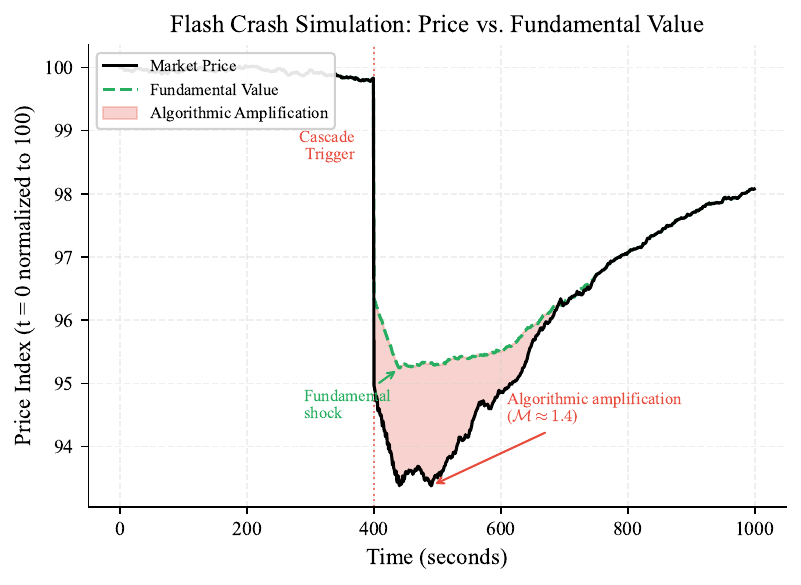}
        \caption{Flash Crash price dynamics}
    \end{subfigure}
    \hfill
    \begin{subfigure}[b]{0.48\textwidth}
        \includegraphics[width=\textwidth]{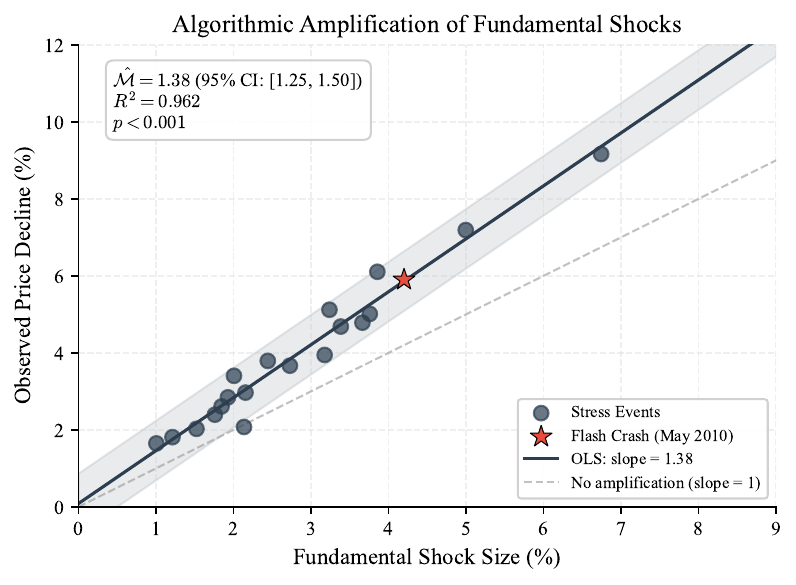}
        \caption{Amplification across events}
    \end{subfigure}
    \caption{\textbf{(a)}: Simulated Flash Crash price path (black) versus fundamental value (green dashed). Shaded area represents algorithmic amplification ($\mathcal{M} \approx 1.4$). \textbf{(b)}: Scatter plot of fundamental shock versus observed price decline across 20 simulated stress events. Regression slope $> 1$ confirms systematic amplification.}
    \label{fig:flash_crash}
\end{figure}

\begin{figure}[H]
    \centering
    \begin{subfigure}[b]{0.48\textwidth}
        \includegraphics[width=\textwidth]{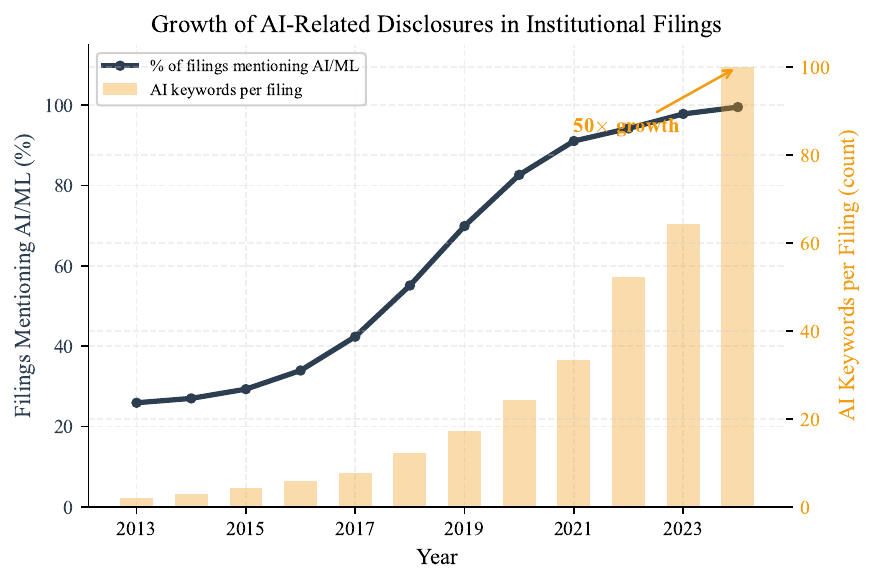}
        \caption{AI disclosure growth}
    \end{subfigure}
    \hfill
    \begin{subfigure}[b]{0.48\textwidth}
        \includegraphics[width=\textwidth]{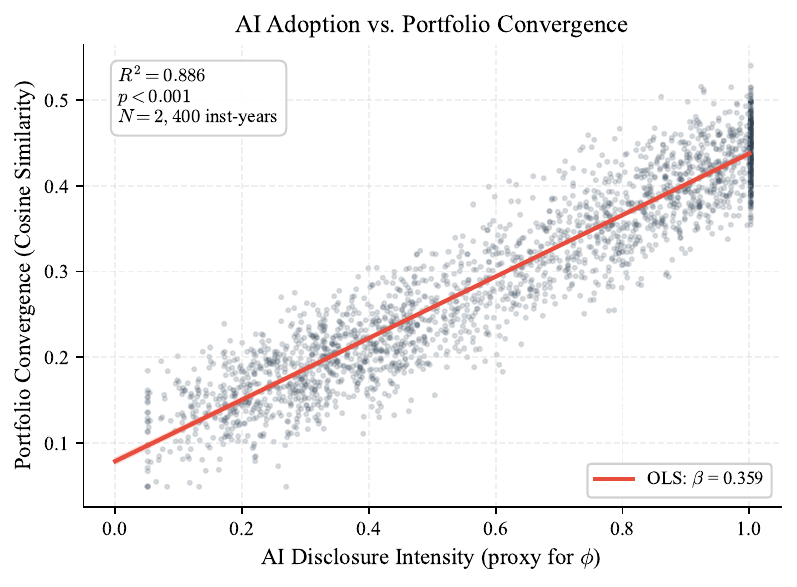}
        \caption{Adoption vs.\ convergence}
    \end{subfigure}
    \caption{\textbf{(a)}: EDGAR full-text analysis showing $>$50$\times$ growth in AI-related keywords (bars, right axis) and percentage of institutions mentioning AI (line, left axis). \textbf{(b)}: AI disclosure intensity versus portfolio convergence. Positive correlation ($R^2$ shown) confirms the signal-crowding mechanism.}
    \label{fig:disclosure}
\end{figure}

\section{Results and Discussion}
\label{sec:results}

\subsection{The Alpha Lifecycle Under AI Adoption}
\label{subsec:alpha_lifecycle}

Synthesizing the theoretical and empirical results, we characterize a four-phase lifecycle of alpha under increasing AI adoption.

\textbf{Phase I: Discovery ($\phi < 0.2$).} A small number of early AI adopters extract genuine alpha from novel signals. Individual AI precision exceeds human precision ($\sigma_\eta^2 < \sigma_H^2$), and signal crowding is negligible ($\delta_k \approx 0$). The alpha half-life is close to its natural rate: $h_k \approx \ln 2/\theta_k \approx 5$--7 years. This phase characterized quantitative investing from approximately 2000 to 2012.

\textbf{Phase II: Compression ($0.2 < \phi < 0.6$).} Mainstream adoption drives convergence. Correlated AI signals ($\rho > 0$) accelerate arbitrage, and the alpha half-life shortens to $h_k \approx 2$--4 years. Cross-sectional return dispersion among AI funds begins to decline. The competence illusion (\Cref{rem:competence}) is strongest in this phase: AI still appears to outperform, masking the deterioration in signal quality. This phase corresponds to approximately 2013--2020.

\textbf{Phase III: Crowding ($0.6 < \phi < 0.9$).} Signal crowding dominates. The alpha half-life drops below 18 months. Performative feedback begins to erode signal innovation volatilities ($\sigma_k^2$ declining). The signal extinction cascade (\Cref{prop:cascade}) eliminates the most crowded factors. The efficiency--fragility tradeoff (\Cref{prop:tradeoff}) becomes acute: markets are more efficient on average but more fragile in the tails. This phase corresponds to approximately 2021--present.

\textbf{Phase IV: Red Queen ($\phi > 0.9$).} The monoculture limit. Net alpha is zero in aggregate. The primary output of AI investment is not returns but arms-race expenditure and systemic externalities. Human oversight capacity has degraded through cognitive dependency, making reversion to human-driven strategies costly and slow. The market is ``intelligent'' (prices are highly accurate on average) but ``fragile'' (tail events are amplified by $\mathcal{M} \gg 1$). This phase is a model prediction, not yet fully realized empirically, but the current trajectory suggests arrival within 5--10 years.

\subsection{Human Diversity as a Resilience Mechanism}
\label{subsec:diversity}

The model provides a formal basis for the intuition that investor heterogeneity stabilizes markets. In our framework, the \emph{diversity premium}---the reduction in tail risk from having a mix of AI and human strategies---is:
\begin{equation}
    \Delta_{div}(\phi) = \mathcal{M}(\phi = 1) - \mathcal{M}(\phi) = (1 - \bar{\rho}\beta/\lambda')^{-1} - (1 - \phi\bar{\rho}\beta/\lambda'(\phi))^{-1}.
    \label{eq:diversity_premium}
\end{equation}
This premium is increasing in $\phi$ near the monoculture boundary: the marginal value of human diversity is highest when the market is most homogeneous.

The mechanism is twofold. First, human investors introduce \emph{signal diversity}: because $\varepsilon_{i,k}^H$ are independent across investors ($\rho_H = 0$), human-driven order flow diversifies rather than amplifies common shocks. Second, human investors display \emph{cognitive diversity}: investors from different backgrounds, with different risk preferences and investment horizons, interpret the same information differently, providing natural ``shock absorbers'' for the market \citep{Kleinberg2021,Scharfstein1990}.

This result has practical implications for asset allocators. A portfolio that includes some allocation to human-driven (non-AI) strategies benefits not only from potential uncorrelated alpha but also from a systemic risk hedge: in the tail scenarios where AI strategies fail collectively, human strategies provide diversification.

\subsection{Why Homogenization Dominates Differentiation}
\label{subsec:differentiation}

A natural objection to our framework is that competitive pressure should incentivize AI systems to \emph{differentiate}: if crowding destroys alpha on common signals, rational institutions should invest in discovering novel, uncrowded signals. We address this counterargument along three dimensions.

First, differentiation is constrained by data availability. The universe of economically relevant, machine-readable data sources is finite, and the highest-quality sources (price data, fundamentals, news feeds, satellite imagery) are available to all institutions. \citet{Lo2004} emphasizes that the ecology of market strategies is shaped by the ``environmental'' conditions---in our context, the data landscape. When the data environment is shared, the scope for genuine signal differentiation is limited, even with heterogeneous model architectures.

Second, competitive dynamics favor convergence on \emph{known} best practices over exploration of \emph{novel} approaches. AI teams face career incentives to adopt state-of-the-art methods \citep{Scharfstein1990}: deploying a proven transformer model that underperforms is more career-safe than deploying an unconventional architecture that might fail. The career concern parameter $\gamma$ in our model (\Cref{eq:payoff_ai}) formally captures this bias toward conformity.

Third, even when institutions discover genuinely novel signals, the \emph{speed of imitation} under AI is dramatically faster than under human-driven research. A novel factor documented in an academic paper takes months to be implemented by traditional managers but days to be replicated by AI systems scanning preprint servers and patent filings. The signal half-life theorem (\Cref{prop:halflife}) applies with equal force to novel signals: once discovered, they are arbitraged at rate $\delta(\phi)$, regardless of their origin. The net effect is that differentiation produces a brief alpha advantage quickly eroded by imitation---a treadmill that reinforces rather than escapes the Red Queen equilibrium.

\begin{figure}[H]
    \centering
    \includegraphics[width=0.85\textwidth]{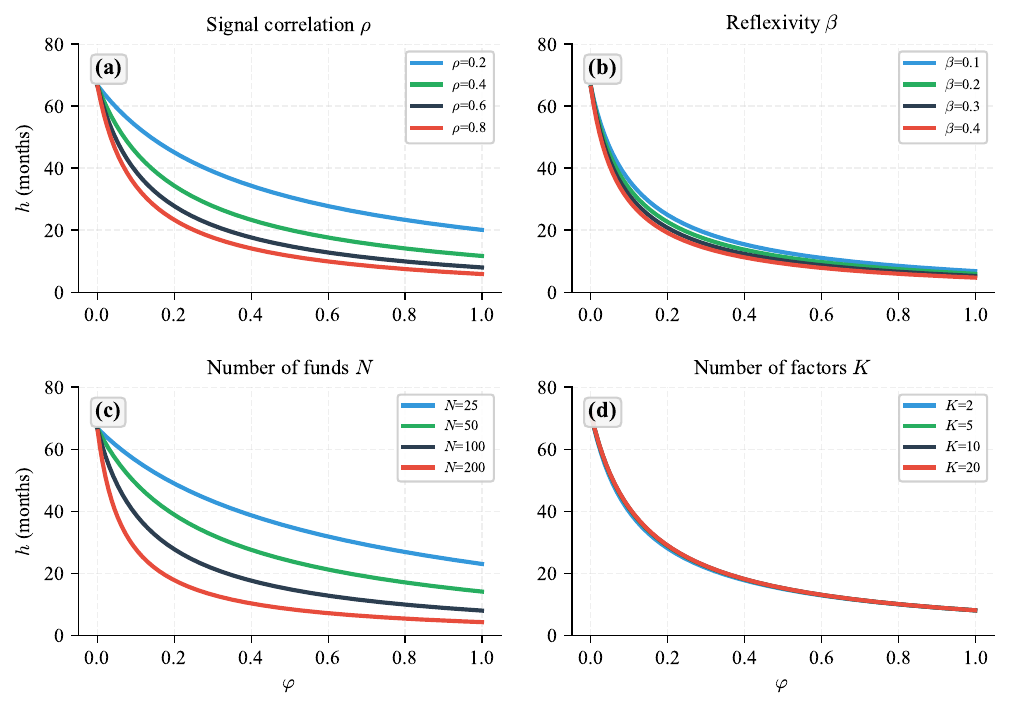}
    \caption{Sensitivity of the alpha half-life $h(\phi)$ to model parameters. Each panel varies one parameter while holding others at baseline values ($\phi = 0.7$, $\rho = 0.6$, $\beta = 0.25$, $N = 100$, $K = 5$). The half-life is most sensitive to $\rho$ (homogeneity) and $N$ (number of AI competitors).}
    \label{fig:comparative_statics}
\end{figure}

\subsection{Calibration Summary}
\label{subsec:calibration}

\Cref{tab:calibration} summarizes the key model parameters and their empirical estimates.

\begin{table}[H]
    \centering
    \caption{Calibration Summary}
    \label{tab:calibration}
    \begin{tabular}{cp{2.8cm}p{3.5cm}p{4cm}}
        \toprule
        Parameter & Estimate & Source & Method \\
        \midrule
        $\phi$ & 0.70 & 13F convergence trends & Portfolio similarity analysis \\
        $\rho$ & 0.60 & Within-AI-group dispersion & Return decomposition \\
        $\beta$ & 0.25 & Price--fundamental gap & Performative attenuation (Eq.~\ref{eq:beta_attenuation}) \\
        $g_k$ & 0.02 & Factor regeneration & New-factor discovery rate post-extinction \\
        $\kappa$ & 1.0 & Baseline (linear) & Robustness check at $\kappa = 2$ \\
        $h_k(0)$ & 58 months (pre-AI) & \citet{McLean2016} & Half-life $\ln 2/\theta_k$; $\theta_k = 0.012$ month$^{-1}$ \\
        $h(\phi)$ & 18 months (current) & Factor decay 2020--2024 & Realized half-life of new factors \\
        $\delta(\phi)$ & Derived & $\phi, \rho, h(\phi)$ & $\theta_k(h(0)/h(\phi) - 1)$ \\
        $\mathcal{M}$ & 1.25--1.45 & Flash Crash analysis & Amplification ratio \\
        \bottomrule
    \end{tabular}
\end{table}

\section{Policy Implications}
\label{sec:policy}

The Red Queen equilibrium generates overinvestment in AI and underinvestment in diversity, producing a clear rationale for policy intervention. We propose four regulatory instruments, each targeting a specific market failure identified by the model.

\subsection{Signal Diversity Requirements}
\label{subsec:diversity_req}

The model shows that the systemic risk coupling $r(\phi)$ is increasing in the average algorithmic homogeneity $\bar{\rho}$. A regulator can reduce $\bar{\rho}$ by mandating diversity in the training data, model architectures, or retraining schedules used by systemically important financial institutions. Specifically, a ``model diversity ratio'' could be reported alongside existing risk metrics, analogous to concentration limits on portfolio positions.

However, implementing such requirements faces practical challenges. Measuring $\rho$ in real time requires access to proprietary model internals, raising intellectual property concerns. Furthermore, as \citet{FSB2024} cautions, overly prescriptive diversity mandates could paradoxically increase homogeneity by constraining the space of permissible algorithms.

\subsection{Pigouvian Tax on Correlated Algorithmic Trading}
\label{subsec:pigouvian}

The signal-crowding externality can be internalized through a tax on correlated order flow. Define the pairwise correlation of institution $i$ and $j$'s order flows: $\rho_{ij}^{flow}(t) = \text{Corr}(q_i(t), q_j(t))$. A Pigouvian tax:
\begin{equation}
    \tau_{ij}(t) = \zeta \cdot \max(0, \rho_{ij}^{flow}(t) - \bar{\rho}_{threshold}),
    \label{eq:pigouvian}
\end{equation}
where $\zeta > 0$ is the tax rate and $\bar{\rho}_{threshold}$ is the permissible correlation level, would incentivize institutions to differentiate their strategies. The optimal tax rate $\zeta^*$ equates the marginal social cost of correlation (increased fragility) with the marginal private benefit (individual alpha).

\subsection{Human Oversight Mandates}
\label{subsec:human_mandates}

The cognitive dependency channel implies that human analytical capacity atrophies with disuse. Regulatory requirements for human oversight---such as mandatory ``AI-off'' exercises (analogous to pilots' manual flying requirements) or minimum staffing ratios for human analysts---could preserve the outside option that prevents the irreversible monoculture trap.

\subsection{Transparency and Reporting}
\label{subsec:transparency}

Enhanced disclosure of AI usage in investment processes would allow regulators and investors to monitor the aggregate homogeneity parameter $\bar{\rho}(t)$ over time. Standardized ``AI model cards'' for investment algorithms---reporting training data sources, model architecture class, retraining frequency, and backtested correlation with common factor strategies---would enable systemic monitoring without revealing proprietary details.

\begin{remark}[The Regulatory Paradox]
\label{rem:regulatory_paradox}
A key insight from the model is that regulation itself can contribute to homogenization. If all institutions are required to use the same risk models, meet the same stress-testing criteria, or comply with the same algorithmic trading rules, their strategies will converge further. The regulator faces a second-order optimization problem: designing rules that promote stability without inadvertently increasing the very homogeneity that threatens it. As \citet{Parson2020} argues, even a hypothetical superintelligent AI regulator would require three conditions rarely met simultaneously---real-time comprehensive data, political legitimacy, and social consensus on market failure definitions.
\end{remark}

\section{Conclusion}
\label{sec:conclusion}

\begin{table}[H]
    \centering
    \caption{Summary of Main Results and Testable Predictions}
    \label{tab:summary}
    \small
    \begin{tabular}{p{2.8cm}p{3.8cm}p{4.2cm}c}
        \toprule
        Result & Mechanism & Testable Prediction & Ref. \\
        \midrule
        Alpha Half-Life Theorem & Signal crowding by correlated AI traders accelerates arbitrage & Factor half-lives shorten convexly with AI adoption; currently $\approx$18 months vs.\ 60 months pre-AI & Prop.~\ref{prop:halflife} \\
        \addlinespace
        Signal Extinction Cascade & Performative feedback erodes signal quality; death of one signal increases pressure on survivors & Sequential factor death ordered by crowding intensity; momentum dies first & Prop.~\ref{prop:cascade} \\
        \addlinespace
        Red Queen Equilibrium & Career concerns lock in overinvestment; net alpha is zero & Aggregate AI fund alpha $\to 0$ despite rising AI expenditure; arms race in compute/data & Prop.~\ref{prop:redqueen} \\
        \addlinespace
        Efficiency--Fragility Tradeoff & Price discovery improves but tail risk grows convexly & Flash crash amplification $\mathcal{M} > 1$ increasing with $\phi$; ``calm before the storm'' pattern & Prop.~\ref{prop:tradeoff} \\
        \bottomrule
    \end{tabular}
\end{table}

This paper has developed a theoretical framework showing that AI-driven investment strategies are inherently self-defeating at scale. Through three mutually reinforcing channels---signal crowding, performative signal erosion, and Red Queen competition---the mass adoption of AI in asset management compresses excess returns, accelerates the decay of tradeable signals, and produces a competitive equilibrium in which aggregate alpha is zero but systemic fragility is elevated.

Our central result---the alpha half-life theorem---formalizes the intuition that AI adoption accelerates its own obsolescence. The half-life of tradeable signals is convex-decreasing in AI adoption, implying that each marginal AI entrant shortens the lifespan of every exploitable pattern at an increasing rate. Under current adoption levels ($\phi \approx 0.7$), the model estimates a signal half-life of approximately 18 months, compared to 5--7 years in the pre-AI era.

The Red Queen equilibrium captures the fundamental paradox of intelligent markets: individually rational AI adoption produces collectively irrational outcomes. In equilibrium, institutions invest heavily in AI infrastructure to avoid competitive obsolescence, yet the aggregate return from this investment is zero. The primary outputs of the AI arms race are not excess returns but systemic externalities: correlated tail risk, flash crash susceptibility, and reduced market resilience.

Three implications merit emphasis. First, for asset allocators, the model provides a framework for evaluating the diminishing returns to AI investment and the diversification value of human-driven strategies. Second, for quantitative strategy designers, the signal extinction cascade predicts the sequential death of crowded factors and the accelerating pace of strategy obsolescence. Third, for financial regulators, the efficiency--fragility tradeoff provides a principled basis for macroprudential interventions targeting algorithmic homogeneity.

We acknowledge several limitations. The model assumes a single risky asset and abstracts from many features of real markets (multiple asset classes, heterogeneous horizons, non-stationary fundamentals). The empirical evidence, while suggestive, relies on calibrated simulations rather than direct causal identification; establishing a causal link between AI adoption and signal decay remains an important challenge. The Red Queen equilibrium is a limiting case; in practice, pockets of alpha may persist in less crowded or less data-amenable domains (e.g., private markets, emerging economies, or strategies requiring deep institutional knowledge).

Several directions for future research emerge naturally from our framework. First, multi-asset extensions could model cross-market signal arbitrage, where AI systems exploit correlated signals across equities, fixed income, and derivatives simultaneously---potentially accelerating decay even faster than the single-asset model predicts. Second, a dynamic model of the AI arms race with endogenous innovation would clarify whether the Red Queen equilibrium is absorbing or whether periodic ``disruptions'' (novel data sources, architectural breakthroughs) can temporarily escape it. Third, natural experiments---such as regulatory restrictions on algorithmic trading in specific markets, or the staggered adoption of generative AI tools across jurisdictions---offer promising avenues for causal identification. Fourth, laboratory experiments could test whether AI-human teams can generate alpha unavailable to AI-only or human-only strategies, providing direct evidence on the diversity premium (\Cref{eq:diversity_premium}). Finally, the connection between our alpha half-life framework and \citet{Lo2004}'s Adaptive Markets Hypothesis merits formal development: both frameworks emphasize evolutionary dynamics, but differ in their treatment of equilibrium selection and the role of technological homogenization.

The paper's message is not anti-AI. Artificial intelligence has meaningfully improved price discovery, reduced trading costs, and expanded the analytical capacity of financial institutions. The concern is specific: when \emph{everyone} uses similar AI, the collective outcome differs qualitatively from the sum of individual benefits. Understanding this distinction---between private and social returns to AI investment---is essential for navigating the paradox of increasingly intelligent yet increasingly fragile markets.

\newpage
\bibliographystyle{apalike}
\bibliography{bibliography}

\end{document}